\documentclass[12pt]{article}%
\usepackage{graphicx}
\usepackage{epstopdf}
\usepackage{amssymb}
\usepackage{amsfonts}
\usepackage{mitpress}%
\usepackage{amsmath}%
\providecommand{\U}[1]{\protect\rule{.1in}{.1in}}
\newtheorem{theorem}{Theorem}
\newtheorem{acknowledgement}[theorem]{Acknowledgement}

\newtheorem{claim}[theorem]{Claim}

\newtheorem{conjecture}[theorem]{Conjecture}
\newtheorem{corollary}[theorem]{Corollary}

\newtheorem{lemma}[theorem]{Lemma}

\newenvironment{proof}[1][Proof]{\noindent\textbf{#1.} }{\ \rule{0.5em}{0.5em}}
\newdimen\dummy
\dummy=\oddsidemargin
\ifx\pdfoutput\relax\let\pdfoutput=\undefined\fi
\newcount\msipdfoutput
\ifx\pdfoutput\undefined\else
\ifcase\pdfoutput\else
\ifx\paperwidth\undefined\else
\ifdim\paperheight=0pt\relax\else\pdfpageheight\paperheight\fi
\ifdim\paperwidth=0pt\relax\else\pdfpagewidth\paperwidth\fi
\fi\fi\fi
\begin{document}

\title{Tiling $R^{5}$ by Crosses}
\author{P. Horak$^{1}$, V. Hromada$^{2}$\\$^{1}$University of Washington, Tacoma, USA\\$^{2}$Slovak University of Technology, Bratislava, Slovakia\\
}
\maketitle

\begin{abstract}
\noindent An $n$-dimensional cross comprises $2n+1$ unit cubes: the center
cube and reflections in all its faces. It is well known that there is a
tiling of $R^{n}$ by crosses for all $n.$ AlBdaiwi and the first author
proved that if $2n+1$ is not a prime then there are $2^{\aleph _{0}}$ \
non-congruent regular (= face-to-face) tilings of $R^{n}$ by crosses, while
there is a unique tiling of $R^{n}$ by crosses for $n=2,3$. They conjectured
that this is always the case if $2n+1$ is a prime. To support the conjecture
we prove in this paper that also for $R^{5}$ there is a unique regular, and
no non-regular, tiling by crosses. So there is a unique tiling of $R^{3}$ by
crosses, there are $2^{\aleph _{0}}$ tilings of $R^{4},$ but for $R^{5}$
there is again only one tiling by crosses. We guess that this result goes
against our intuition that suggests "the higher the dimension of the \
space, the more freedom we get".

\end{abstract}

\noindent Tilings of $R^{n}$ by unit cubes go back to 1907 when Minkowski
conjectured \cite{M} that each lattice tiling of $R^{n}$ by unit cubes
contains twins, a pair of cubes sharing a complete $n-1$ dimensional face.
This conjecture was proved by Haj\'{o}s \cite{Hajos} in 1942.\bigskip\ 

\noindent In 1930, when Minkowski's conjecture was still open, Keller
\cite{Keller} suggested that the lattice condition in the conjecture is
redundant, that the nature of the problem is purely geometric, and not
algebraic as assumed by Minkowski. Thus he conjectured that each tiling of
$R^{n}$ by unit cubes contains twins. It is trivial to see that each tiling of
$R^{2}$ by unit cubes contains twins, and it is also easy to verify it for
$R^{3}.$ However, a proof that each tiling of $R^{n},$ $4\leq n\leq6,$
contains twins takes in aggregate 80 pages, see \cite{Peron}. There was no
progress on Keller's conjecture for more than 50 years. Only in 1992 Lagarias
and Shor \cite{LS} constructed a tiling of $R^{n},n\geq10,$ by unit cubes with
no twins. First they found such a tiling in $R^{10},$ which we consider a very
surprising and remarkable result. However, once one has such a tiling in hand,
it is relatively easy to find it for $R^{n},n>10,$ as well. The second part
supports our belief that "the higher the dimension of the space, the more
freedom we get". Mackey \cite{Mackey} proved that the Keller's conjecture is
false for $n=8,9$ as well. As to the remaining value of $n=7,$ there are only
some partial results, see \cite{D}.\bigskip

\noindent Since late fifties tilings of $R^{n}$ by different clusters of unit
cubes have been considered, see e.g. \cite{Stein} and \cite{U}, many of them
related to perfect error-correction codes in Lee metric (also called Manhattan
metric in $Z^{n}$). The Golomb-Welch conjecture \cite{GW} has been a main
motivating power of the research in this area for the last forty years. A
perfect $e$-error correcting Lee code over $Z$ of block size $n$, denoted
$PL(n,e)$, is a set $C\subset Z^{n}$ of codewords so that each word $A\in
Z^{n}$ is at Lee distance at most $e$ from exactly one codeword in $C.$
Similarly, a perfect $e$-error correcting Lee code over $Z_{q}$ of block size
$n$, denoted $PL(n,q,e)$, is a set $C\subset Z_{q}^{n}$ of codewords so that
each word $A\in Z_{q}^{n}$ is at Lee distance at most $e$ from exactly one
codeword in $C.$

\begin{conjecture}
Golomb-Welch. For $n\geq 3$ and $e>1,$ there is no $PL(n,e)$ code.
\end{conjecture}

\noindent Clearly, the above conjecture, if true, implies that there is no
$PL(n,q,e)$ code for $n\geq3$, $e>1,$ and $q\geq2e+1.$ For the state of the
art on the conjecture we refer the reader to \cite{H1}.\bigskip

\noindent In this paper we focus on tilings by $n$-crosses. An $n$-dimensional
cross comprises $2n+1$ unit cubes: the "central" one and reflections in all
its faces. A tiling $\mathcal{L}$ of $R^{n}$ by crosses is called a $Z$-tiling
if centers of all crosses in $\mathcal{L}$ have integer coordinates. Further,
$\mathcal{L}$ is called a lattice tiling if centers of all crosses in
$\mathcal{L}$ form a lattice. A regular (also called a face-to-face) tiling is
a tiling that is congruent to a $Z$-tiling; otherwise the tiling is called
non-regular. We recall that two tilings $\mathcal{T}$ and $\mathcal{S}$ of
$R^{n}$ are congruent if there exists a linear, distance preserving bijection
of $R^{n}$ which maps $\mathcal{T}$ on $\mathcal{S}$.\ It seems that
K\'{a}rteszi \cite{Kar} was the first to ask whether there exists a tiling of
$R^{3}$ by crosses. Such a tiling was constructed by Freller in 1970;
Korchm\'{a}ros about the same time treated the case $n>3$. Golomb and Welch
showed the existence of these tilings in terms of error-correcting codes, see
Section 3.5 in \cite{Stein}. Immediately after the existence question has been
answered, the enumeration of tilings has been studied. \ In \cite{Mol}
Moln\'{a}r proved:

\begin{theorem}
\label{M}Molnar. The number of pair-wise non-congruent lattice $Z$-tilings of
$R^{n}$ by crosses equals the number of non-isomorphic Abelian groups of order
$2n+1.$
\end{theorem}

\noindent Szab\'{o} \cite{S} constructed a non-regular lattice tiling of
$R^{n}$ by crosses in the case when $2n+1$ is not a prime. Using refinements
of this construction it was proved in \cite{H} that in this case there are
$2^{\aleph_{0}\text{ }}$non-congruent $Z$-tilings of $R^{n}$ by crosses. In a
strict contrast to this result it was proved there that, for $n=2,$ and $n=3$,
there is a unique, up to a congruence, tiling of $R^{n}$ by crosses. It is
conjectured in \cite{H}, see also \cite{Bader}:

\begin{conjecture}
\label{2}If $2n+1$ is a prime then there exists, up to a congruence, only one
$Z$-tiling of $R^{n}$ by crosses.
\end{conjecture}

\noindent It seems to us that the above conjecture, if true, would totally go
against our intuition that suggests: the higher the dimension of the space,
the more freedom we get; see also an above comment related to the
Lagarias-Shor result on Keller's conjecture.\bigskip\ 

\noindent To provide supporting evidence for Conjecture \ref{2} we prove in
this paper:

\begin{theorem}
\label{1}There exists, up to a congruence, a unique $Z$-tiling of $R^{5}$ by crosses.
\end{theorem}

\noindent We note that a sketch of a proof of the above statement has been
given in \cite{HH1}. However, the sketch is so short that it is impossible for
the interested reader to reconstruct the whole proof from it. Therefore in
this paper a complete version of the proof is provided. Although we proved
Conjecture \ref{2} only for $n=5,$ an essential part of the proof of Theorem
\ref{1} holds for all $n=2(\operatorname{mod}3).$ We believe that this part
will be helpful when proving this conjecture for some other values of
$n.$\bigskip

\noindent Clearly, if $\mathcal{L}$ is a $Z$-tiling of $R^{n}$ by crosses,
then centers of crosses in $\mathcal{L}$ form a $PL(n,1)$ code. It is easy to
check that the unique tiling of $R^{5}$ by crosses is $11$-periodic. Thus, as
an immediate consequence we get:

\begin{corollary}
There is a $PL(5,q,1)$ code if and only if $11|q.$
\end{corollary}

\noindent As to the non-regular tilings of $R^{n}$ by crosses, it was
mentioned above that such a tiling exists if $2n+1$ is not a prime. A result
of Redei \cite{Redei} implies that, if $2n+1$ is a prime, then there is no
lattice non-regular tiling of $R^{n}$ by crosses. It is easy to check that a
non-regular tiling of $R^{2}$ by crosses does not exist. The same result for
$n=3$ has been proved in \cite{G}. As the other main result of this paper we
will show that:

\begin{theorem}
\label{3}Let $2n+1$ be a prime. If there is a unique $Z$-tiling of $R^{n}$ by
crosses, then there is no non-regular tiling of $R^{n}$ by crosses.
\end{theorem}

\noindent Combining Theorem \ref{1} with \ref{3} we get:

\begin{corollary}
There is a unique, up to a congruence, tiling of $R^{5}$ by crosses, and this
tiling is a $Z$-tiling.
\end{corollary}

\noindent Thus, there is a unique tiling of $R^{3}$ by crosses, there are
$2^{\aleph_{0}\text{ }}$ pair-wise non-congruent $Z$-tilings of $R^{4\text{ }%
}$by crosses, but for $R^{5}$ there is again a unique tiling by
crosses.\bigskip

\noindent Also, by means of Theorem \ref{3}, it is straightforward that
Conjecture \ref{2} is equivalent to

\begin{conjecture}
If $2n+1$ is a prime then there exists, up to a congruence, a unique tiling of
$R^{n}$ by crosses, and this tiling is a lattice $Z$-tiling.
\end{conjecture}

\noindent In the next section we introduce needed notation, definitions and
state some auxiliary results. Theorem \ref{3} will be proved in Section 2, while Theorem \ref{1} will be proved in Section 3.

\section{Preliminaries}

\noindent In this section we recall some notations, notions, and results
which will turn out to be useful in proving both main results of the paper,
Theorem \ref{1} and Theorem \ref{3}.\bigskip 

\noindent Since the problem of tilings by crosses comes originally from the
area of error-correcting codes we will stick to some of its terminology. Let
$\mathcal{L}$ be a $Z$-tiling of $R^{n}$ by crosses. We will denote by
$\mathcal{T}_{\mathcal{L}}\mathcal{\subset}Z^{n}$ the set of centers of
crosses in $\mathcal{L}.$ The elements of $Z^{n}$ will be called words while
the words in $\mathcal{T}_{\mathcal{L}}$ will be called codewords. We will
also say that a codeword $W$ covers a word $V$ if $\rho_{M}(V,W)\leq1.$ As
usual $\rho_{M}$ stands for the Manhattan distance of $V=(v_{1},...,v_{n})$
and $W=(w_{1},...,w_{n})$ given by
\[
\rho_{M}(V,W)=%
%TCIMACRO{\dsum \limits_{i=1}^{n}}%
%BeginExpansion
{\displaystyle\sum\limits_{i=1}^{n}}
%EndExpansion
\left\vert v_{i}-w_{i}\right\vert .
\]
The weight $\left\vert V\right\vert _{M}$ of $V\in Z^{n}$ is given by
$\left\vert V\right\vert _{M}:=$ $%
%TCIMACRO{\dsum \limits_{i=1}^{n}}%
%BeginExpansion
{\displaystyle\sum\limits_{i=1}^{n}}
%EndExpansion
v_{i}=\rho_{M}(V,O),$ where $O=(0,...,0).$ \ The following simple observation
will be used several times:

\begin{claim}
\label{0} Let $\mathcal{L}$ be a tiling of $R^{n}$ by crosses. Then permuting
the order of coordinates of each codeword in $\mathcal{T}_{\mathcal{L}}$
and/or changing a sign of a coordinate for each codeword in $T_{\mathcal{L}}$
and/or adding $\ $a word $V\in R^{n}$ to each codeword results in a set
$\mathcal{T}^{\prime}$ which \ induces a tiling of $R^{n}$ by crosses
congruent to $\mathcal{L}$.
\end{claim}

\noindent If $\mathcal{L}$ is a tiling of $R^{n}$ by crosses then for each
word $V$ in $Z^{n}$ there is a unique codeword $W$ in $\mathcal{T}_{\mathcal{%
L}}$ so that $\rho _{M}(V,W)\leq 1.$ Therefore $\mathcal{T}_{\mathcal{L}}$
can be also seen as a decomposition (tiling) of $Z^{n}$ by Lee spheres $%
S_{n,1}$ of radius $1$ centered at $O$, where $S_{n,1}=\{V\in Z^{n},\rho
_{M}(V,O)\leq 1\}=\{O\}\cup \{e_{i},i=1,...,n\};$ and vice versa, each
tiling of $Z^{n}$ by spheres $S_{n,1}$ induces a tiling of $R^{n}$ by
crosses. As usual, $e_{i\text{ }}=(0,...,0,1,0,...,0)$ where the $i$-th
coordinate equal to $1.$\bigskip 

\noindent In general, if $S$ is a subset of $R^{n}$ $(Z^{n}),$ a tiling $%
\mathcal{L}$ of $R^{n}$ $(Z^{n})$ by translations of $S$ can be described in
the form $\{S+\mathbf{u,u}\in \mathcal{U}\},$ where $\mathbf{u}$ is a
vector. Then $\mathcal{L}$ is a lattice tiling if $\mathcal{U}$ is a
lattice.\ For the sake of simplicity we will abuse slightly the language and
a subset $\mathcal{U}$ of $R^{n}$ $(Z^{n})$ will be understood sometimes as
a set of vectors with the obvious $U\in \mathcal{U}$ meaning that the vector 
$\mathbf{u=}U-O$ is in $\mathcal{U}.$ The following theorem stated in \cite%
{H1} turns out to be useful when proving both main results of the paper.

\begin{theorem}
\label{14} Let $S$ be a subset of $Z^{n}.$ Then there is a lattice tiling of 
$Z^{n}$ by translations of $S$ if and only if there is an abelian group $G$
of order $\left\vert S\right\vert $ and a homomorphism $\phi
:Z^{n}\rightarrow G$ so that the restriction of $\phi $ to $S$ is a
bijection. In addition, if $\phi $ satisfies this condition, then the
lattice tiling of $Z^{n}$ by translations of $S$ is given by $\{S+\mathbf{u,u%
}\in \ker (\phi )\}$.
\end{theorem}

\noindent As an immediate consequence we get:

\begin{corollary}
\label{20} Let $\phi :Z^{n}\rightarrow Z_{2n+1}$, the cyclic group of order $%
2n+1,$ be a homomorphism so that, for all $1\leq i<j\leq n,$ $\phi (e_{i})$
is not an inverse element to $\phi (e_{j}),$ that is $\phi (e_{i})\neq -\phi
(e_{j}).$ Then $\{S_{n,1}+\mathbf{u,u}\in \ker \phi \}$ is a lattice tiling
of $Z^{n}$ by $S_{n,1}.$
\end{corollary}

\noindent We note that tiling of $R^{n}$ by crosses given in \cite{H1} and
other papers is a lattice tilling. Therefore these tilings can be seen as
obtained by Corollary \ref{20}.\bigskip

\noindent Let $\mathcal{L}$ be a collection of crosses that tile $R^{n}.$ We
will always assume wlog that the cross $K_{O}$ centered at the origin
belongs to $\mathcal{L}$. Then each cross $K\in \mathcal{L}$ can be seen as
a translation of $K_{O}$ by a vector $\mathbf{u}$. So $\mathcal{L}=\{K_{O}+%
\mathbf{u,u}\in \mathcal{T}_{\mathcal{L}}\}.$ For the sake of brevity we
will use $K_{\mathbf{u}}$ for a cross centered at a point $U=O+\mathbf{u}.$

\section{Proof of Theorem \protect\ref{3}}

\noindent In this section we provide a proof of Theorem \ref{3}. The
following lemma will be the key ingredient of the proof. We recall that by
Theorem \ref{M} there is a unique lattice tiling of $R^{n}$ by crosses when $%
2n+1$ is a prime.

\begin{lemma}
\label{LL}Let $2n+1$ be a prime, and let $\mathcal{D}$ be a unique lattice
tiling of $R^{n}$ by crosses. If $\mathcal{K}$ is a cross in $\mathcal{D}$,
then shifting $\mathcal{K}$ along any axis will cause that all crosses of $%
\mathcal{D}$ will be shifted as well.
\end{lemma}

\begin{proof}
As $\mathcal{D}$ is a lattice tiling it suffices to prove the statement for
the cross $K_{O}.$\bigskip

\noindent Consider the homomorphism $\phi :Z^{n}\rightarrow Z_{2n+1}$ given
by $\phi (\mathbf{e}_{i})=i$ for all $i=1,...,n.$ Then, by Corollary \ref{20}%
, $\phi $ induces a lattice tiling $\mathcal{D}=\mathcal{\{}S_{n,1}+\mathbf{u%
}$\textbf{$,$}$\mathbf{u}\in \ker (\phi )\}$ of $R^{n}$ by crosses.
\noindent Let $j$, $1\leq j\leq n,$ be fixed. We will prove that shifting
the cross $K_{O}$ along the $j$-th axis would shift all crosses in $\mathcal{%
D}$. We start with describing vectors $\mathbf{v}_{1}\mathbf{,...,v}_{n}$
that form a basis of the lattice $\ker (\phi ).$ Let $j^{-1}$ be the element
inverse to $j$ in the multiplicative abelian group $Z_{2n+1}^{\ast }.$ For
each $i=1,...,n,$ $i\neq j,$ we set $\mathbf{v}_{i}=\mathbf{e}_{i}-ij^{-1}%
\mathbf{e}_{j},$ and $\mathbf{v}_{j}=(2n+1)\mathbf{e}_{j}.$ Clearly, $\phi (%
\mathbf{v}_{i})=0,$ that is, $\mathbf{v}_{i}\in \ker (\phi ).$ Indeed, for $%
i\neq j,$ $\phi (\mathbf{v}_{i})=\phi (\mathbf{e}_{i}-ij^{-1}\mathbf{e}%
_{j})=\phi (\mathbf{e}_{i})-ij^{-1}\phi (\mathbf{e}_{j})=(i-ij^{-1}j)\text{mod}
(2n+1)=0,$ and $\phi (\mathbf{v}_{j})=\phi ((2n+1)\mathbf{e}_{j})=(2n+1)j%
\text{mod}(2n+1)=0.$ Let $A$ be the matrix whose rows are vectors $\mathbf{v}%
_{1}\mathbf{,...,v}_{n}.$ It is easy to calculate $\det A$ as the rows and
columns of $A$ can be permuted such that the resulting matrix is a lower
triangular having $(2n+1,1,1,...,1)$ as its diagonal entries. Therefore, $%
\det A=2n+1,$ which in turn implies that $\mathbf{v}_{1}\mathbf{,...,v}_{n}$
form a basis of the lattice $\ker (\phi ).$\bigskip\ 

\noindent Assume that the cross $K_{O\text{ }}$has been shifted along the $j$%
-th axis. Then this will cause that the cross $K_{\mathbf{v}_{i}}$, $%
i=1,...,n,$ will be shifted as well. Indeed, for $i\neq j,$ the cross $K_{%
\mathbf{v}_{i}}$ contains the unit cube $C_{i}$ centered at $\mathbf{v}_{i}-%
\mathbf{e}_{i}=-ij^{-1}\mathbf{e}_{j}$ (centered at $2n\mathbf{e}_{j}$ for $%
i=j$); that is, the center of $C_{i}$ lies on $j$-th axis. Further, the
cross $K_{O}$ contains the cube $C_{O}$ centered at $O.$ Thus, when shifting 
$K_{O}$ along the $j$-th axis we shift the cube $C_{O}$ along this axis, and
this will cause the cube $C_{i}$ to get shifted; i.e., the cross $K_{\mathbf{%
v}_{i}}$ will be shifted along the $j$-th axis for all $i=1,...,n$. Consider
now a cross $K_{\mathbf{u}}$ in $\mathcal{D}.$ As $\mathcal{D}$ is a lattice
tiling, the above proved statement is true for any cross $K_{\mathbf{u}}.$
Hence:\bigskip

\noindent \textbf{Claim A.} Shifting the cross $K_{\mathbf{u}}$ along the $j$%
-th axis will cause shifting the cross $K_{\mathbf{u+v}_{i}}$ for all $%
i,1\leq i\leq n.$\bigskip

\noindent With this claim in hand it is easy to provide the closing argument
of our proof. Let $K_{\mathbf{u}}\in \mathcal{D}$. We will prove that
shifting the cross $K_{O}$ along the $j$-th axis will cause that the cross $%
K_{\mathbf{u}}$ will be shifted as well. Since $K_{\mathbf{u}}\in \mathcal{D}
$, it is $\mathbf{u}\in \ker (\phi ),$ and because $\mathbf{v}_{1}\mathbf{%
,...,v}_{n}$ form a basis of $\ker (\phi ),$ $\mathbf{u}$ can be written as
a linear combination $\mathbf{u}=\alpha _{1}\mathbf{v}_{1}+...+\alpha _{n}%
\mathbf{v}_{n},$ where $\alpha _{i}\in Z$ for all $i.$ So to finish the
proof it suffices to apply repeatedly Claim A.\bigskip
\end{proof}

\noindent Now we are ready to prove Theorem \ref{3}.\bigskip

\begin{proof}
\textbf{of Theorem }\ref{3}. Let $\mathcal{L}$ $=\{K_{O}+\mathbf{u}$\textbf{$%
,$}$\mathbf{u}\in \mathcal{U}\}$ be a non-regular tiling of $R^{n}$. Then
there is $i,1\leq i\leq n,$ and a vector $\mathbf{u}=(u_{1},...,u_{n})\in 
\mathcal{U}$ such that $u_{i}$ is not an integer. Let $\alpha \in (0,1)$ be
the fractional part of $u_{i}$. \ Denote by $\mathcal{U}_{\alpha }^{i}$ the
set of all vectors $\mathbf{v}=(v_{1},...,v_{n})$ in $\mathcal{U}$ such that 
$v_{i}-\left\lfloor v_{i}\right\rfloor =\alpha .$ It is known, see e.g. \cite%
{Stein}, that the collection of crosses $K_{\mathbf{u}}\mathbf{,u\in }%
\mathcal{U}_{\alpha }^{i}$ forms a prism $\mathcal{P}$ along the $i$-th
axis; i.e., if a point $X\in \mathcal{P}$ then, for all $c\in R,$ also the
point $X+c\mathbf{e}_{i}\in \mathcal{P}.$ Hence, shifting all crosses $K_{%
\mathbf{v}}\mathbf{,v\in }\mathcal{U}_{\alpha }^{i}$ by any vector $\mathbf{w%
}$ parallel to $\mathbf{e}_{i},$ independently on other crosses in $\mathcal{%
L}$, results in a new tiling of $R^{n}$ by crosses, see e.g. \cite{Stein} or 
\cite{S}. Moreover, if $\mathbf{w}=(m-\alpha )\mathbf{e}_{i},m\in Z,$ then
the shift results in a tiling where all crosses $K_{\mathbf{v}}\mathbf{,v\in 
}\mathcal{U}_{\alpha }^{i}$ are now centered at points with the $i$-th
coordinate being an integer. Repeatedly applying this procedure to other
crosses that have a non-integer coordinate, we arrive at a $Z$-tiling $%
\mathcal{L}^{\ast }$ of $R^{n}$ by crosses. Since we have started with a
non-regular tiling $\mathcal{L}$, there is a proper subset $\mathcal{C\ }$of 
$\mathcal{L}^{\ast }$ of crosses so that $\mathcal{C}$ comprises a prism
along one of the axis.\bigskip

\noindent By Lemma \ref{LL}, if the lattice tiling $\mathcal{D}$ contains a
prism along any axis, this prism constitutes all crosses in $\mathcal{D}$.
Therefore the above tiling $\mathcal{L}^{\ast }$ is not congruent to the
tiling $\mathcal{D}$. However, this contradicts our assumption that there is
a unique $Z$-tiling of $R^{n}$ by crosses. The proof of Theorem \ref{3} is
complete.
\end{proof}

\section{Proof of Theorem \ref{1}}

\noindent Let $\mathcal{L}$ be a $Z$-tiling of $R^{n}$ by crosses, and let
$\mathcal{T}_{\mathcal{L}}\mathcal{\subset}Z^{n}$ be the set of centers of
crosses in $\mathcal{L}$. Since we will deal only with $Z$-tilings by crosses
most of the time we will drop $Z$- and refer to $\mathcal{L}$ as a tiling of
$R^{n}$ by crosses. We use the terminology of coding theory; that is, the
elements of $Z^{n}$ will be called words and the elements of $\mathcal{T}%
_{\mathcal{L}}$ will be called codewords. In
this section we provide a complete proof of Theorem \ref{1}.\bigskip

\noindent As mentioned in the introduction Molnar \cite{Mol} proved that the
number of non-congruent lattice tilings of $R^{n}$ by crosses equals the
number of non-isomorphic abelian groups of order $2n+1.$ As $2n+1$ is a
prime for $n=5,$ there is only one abelian group of order $11,$ and thus
there is a unique, up to congruence, lattice tiling of $R^{5}$ by crosses.
Thus, to prove the main result it suffices to show:

\begin{theorem}
\label{33}Let $\mathcal{L}$ be a tiling of $R^{5}$ by crosses. Then $\mathcal{%
L}$ is a lattice tiling.
\end{theorem}

\noindent Let $W$ be a codeword in
$\mathcal{T}_{\mathcal{L}}$. Then $N_{k}(W)$, the $k$-neighborhood of $W,$
will be the set of codewords $V$ in $\mathcal{T}_{\mathcal{L}}$ at the
distance at most $k$ from $W,$ that is, $N_{k}(W)=\{V\in\mathcal{T}%
_{\mathcal{L}},$ $\rho_{M}(W,V)\leq k\}.$ In the case of $W = O$, we will write $N_{k}$ instead of $N_{k}(O)$. We will say that two $k$%
-neighborhoods $N_{k}(W)$ and $N_{k}(W^{\prime})$ are equal if $\{V-W,V\in
N_{k}(W)\}=\{V-W^{\prime},V\in N_{k}(W^{\prime})\};$ and we will say that
$N_{k}(W)$ and $N_{k}(W^{\prime})$ are congruent if there is a linear distance
preserving transformation mapping $N_{k}(W)$ on $N_{k}(W^{\prime})$. Clearly,
for each codeword $W$, the neighborhoods $N_{1}(W),$ and $N_{2}(W) $ are empty
sets. \bigskip

\noindent The proof of Theorem \ref{33} will be based on:

\begin{theorem}
\label{4}Let $\mathcal{L}$ be a tiling of $R^{5}$ by crosses. Then, for each
codeword $W$ in $\mathcal{T}_\mathcal{L}$, \ the neighborhood $N_{3}(W)$ and $N_{3}(O)$
are equal, and $N_{3}(O)$ is symmetric; that is, if $W\in N_{3}(O)$ then $%
-W\in N_{3}(O)$ as well.
\end{theorem}

\noindent Now we show that the above theorem implies Theorem \ref{33}.\bigskip

\begin{proof}
\textbf{of Theorem \ref{33}}. To show that $\mathcal{L}$ is a lattice tiling
it suffices to prove that, for all codewords $W,Z\in \mathcal{T}_\mathcal{L}$, $W-Z\in 
\mathcal{T}_\mathcal{L}$ as well. As $\mathcal{L}$ is a tiling by crosses, it is not
difficult to see that, for each codeword $Z\in \mathcal{T}_\mathcal{L}$, there is a
sequence $Z_{0}=O,Z_{1},...,Z_{m-1,}Z_{m}=Z$ of codewords in $\mathcal{T}_\mathcal{L}$
such that $\rho _{M}(Z_{i-1},Z_{i})=3,i=1,...,m.$ Then $Z_{i}\in
N_{3}(Z_{i-1})$ and because, by Theorem \ref{4}, the -neighborhoods $%
N_{3}(Z_{i-1})$ and $N_{3}(O)$ are equal, which in turn implies, again by
Theorem \ref{4}, that $-U_{i}\in \mathcal{T}_\mathcal{L}$ as well for all $i=1,...,m.$
Repeatedly applying Theorem \ref{4} we get that $W-U_{1},W-U_{1}-U_{2},$ $%
...,W-U_{1}-U_{2}-...-U_{m}=$ $%
W-(Z_{1}-O)-(Z_{2}-Z_{1})-...-(Z_{m-1}-Z_{m-2})-(Z_{m}-Z_{m-1})=W-Z_{m}=W-Z$
is a codeword. The proof of Theorem \ref{33} is complete.\bigskip
\end{proof}

\noindent Hence, to prove the main result it suffices to prove Theorem \ref%
{4}. It turns out that in order to be able to do so one needs to look at
"wider" neighbourhoods. In fact, to be able to prove Theorem \ref{4} we will
have to prove the same type of a theorem for $5$-neighbourhoods. This
recalls a situation when one wants to prove a statement $P$ by using
mathematical induction, but to be able to prove the inductive step a
statement stronger than $P$ has to be proved.

\begin{theorem}
\label{00}Let $\mathcal{L}$ be a tiling of $R^{5}$ by crosses. Then, for
each $W$ in $\mathcal{T}_\mathcal{L}$, \ the neighborhood $N_{5}(W)$ and $N_{5}(O)$ are
equal, and $N_{5}(O)$ is symmetric.
\end{theorem}

\noindent We will do it in four steps. To
facilitate our discussion we introduce more notation and terminology. By a
word of type $[m_{1}^{\alpha_{1}},...,m_{s}^{\alpha_{s}}]$ we mean a word
having $\alpha_{1}$ coordinates equal to $\pm m_{1},$ ..., $\alpha_{s}$
coordinates equal to $\pm m_{s},$ the other coordinates equal to 0. E.g.,
both words $(-2,-2,-1,-2,0,0)$ and $(1,0,2,0,-2,2)$ are of type $\left[
2^{3},1^{1}\right]  .$ There are three types of words $V$ with its
weight$\left\vert V\right\vert _{M}=3;$ either $V$ is of type $\left[
3^{1}\right]  , $ or of type $\left[  2^{1},1^{1}\right]  ,$ or of type
$\left[  1^{3}\right]  .$ Let $Z\in N_{k}(W).$ Then $Z$ will be called a
codeword of a type with respect to $W$ if $Z-W$ is of the given type; the
number of codewords of type $[m_{1}^{\alpha_{1}},...,m_{s}^{\alpha_{s}}]$ in
$N_{k}(W)$ will be denoted $\left\vert [m_{1}^{\alpha_{1}},...,m_{s}%
^{\alpha_{s}}]\right\vert _{W}$. If the codeword $W$ will be clear from the
context, we will drop the subscript $W.$ Similarly, each word $V,\left\vert
V\right\vert _{M}=4,$ is either of type $\left[  4^{1}\right]  ,$ or $\left[
3^{1},1^{1}\right]  ,$ or $\left[  2^{2}\right]  ,$ or $\left[  2^{1}%
,1^{2}\right]  ,$ or $\left[  1^{4}\right]  .$ \bigskip

\noindent Now we are ready to describe the four phases of proving Theorem %
\ref{4}.\bigskip

(A) Let $\mathcal{L}$ $\ $be a tiling of $Z^{n}$ by crosses. First we prove a quantitative statement, which will be proved not only for $n=5\,\ $but for all \mbox{$n=2(\operatorname{mod}3\dot{)}.$} We believe that this
statement might turn to be very useful when proving Conjecture \ref{2} for
other values of $n$, where $2n+1$ is a prime. Let $W$ be a codeword. The
statement claims that the number of codewords of type $[m_{1}^{\alpha_{1}%
},...,m_{s}^{\alpha_{s}}],$ where $%
%TCIMACRO{\dsum \limits_{i=1}^{s}}%
%BeginExpansion
{\displaystyle\sum\limits_{i=1}^{s}}
%EndExpansion
\alpha_{i}m_{i}\leq4,$ with respect to $W$ depends only on $n$ and does not
depend on $\mathcal{L}.$

\begin{theorem}
\label{5}Let $\mathcal{L}$ be a tiling of $R^{n}$ by crosses where
$n=2(\operatorname{mod}3\dot{)}$ and $W$ be a codeword$.$ Then the number of
codewords of given type with respect to $W$ is: $\left\vert [3^{1}]\right\vert
_{W}=0,$ $\left\vert [2^{1},1^{1}]\right\vert _{W}=2n,$ and $\left\vert
[1^{3}]\right\vert _{W}=\frac{2n(n-2)}{3}.$ Further, $\left\vert
[4^{1}]\right\vert _{W}=\left\vert [2^{2}]\right\vert _{W}=0,\left\vert
[3^{1},1^{1}]\right\vert _{W}=2n,\left\vert [2^{1},1^{2}]\right\vert
_{W}=2n(n-2),$ and $\left\vert [1^{4}]\right\vert _{W}=\frac{n(n-2)(n-3)}{3}.
$
\end{theorem}

(B) We prove an analogue of Theorem \ref{5} for the number of codewords of
type $[m_{1}^{\alpha_{1}},...,m_{s}^{\alpha_{s}}],$ where $%
%TCIMACRO{\dsum \limits_{i=1}^{s}}%
%BeginExpansion
{\displaystyle\sum\limits_{i=1}^{s}}
%EndExpansion
\alpha_{i}m_{i}\leq5.$ However, we get the explicit values for the number of
codewords of individual types only for $n=5,$ while for
$n=2(\operatorname{mod}3)$ we get those values only as a function of the
number of codewords of type $\left[  5^{1}\right]  .$ We point out, that this
is not because the methods used are not satisfactory but for some values
$n=2(\operatorname{mod}3)\,,$ say $n=62,$ there are two (lattice) tilings of
$Z^{n}$ by crosses with different number of codewords of type $\left[
5^{1}\right]  .$ We stress that for $n=62,$ the number $2n+1=125$ is not a
prime, hence it does not provide a counterexample to our conjecture.\bigskip

(C) In this phase we prove that for any two codewords in $\mathcal{T}%
_{\mathcal{L}}$ their $5$-neighborhoods are congruent.\bigskip\ 

(D) As the last step we show that for any two codewords in $\mathcal{T}%
_{\mathcal{L}}$ their $5$-neighborhoods are not only congruent but the two
$5$-neighborhoods equal, and this joint neighborhood is symmetric, so we prove
Theorem \ref{00}.

\subsection{Phase A}

\noindent In this subsection we prove Theorem \ref{5}. In fact we prove an
extended version of the statement. \bigskip

\noindent For any codeword $W$ in $\mathcal{T}_\mathcal{L}$ there are $2n$ words $V$ of
type $\left[  2^{1}\right]  $ with respect to $W$. (We recall that this means
that $V-W$ is of given type). Each of them is covered by a codeword of type
$\left[  3^{1}\right]  ,$ or by a codeword of type $\left[  2^{1},1^{1}\right]
,$ with respect to $W.$ On the other hand, each codeword of type $\left[ 3^{1}\right]  $ and of type $\left[  2^{1},1^{1}\right]  ,$ with respect to $W,$
covers exactly one word of type $\left[  2^{1}\right]  $ with respect to $W.$
Thus we get, for each codeword $W,$%

\begin{equation}
\left\vert \lbrack 3^{1}]\right\vert +\left\vert [2^{1},1^{1}]\right\vert =2n
\label{p}
\end{equation}

\noindent The above and the following equalities are valid for each codeword
$W,$ therefore in what follows we drop the index $_{W}.$ Also we will not
repeat any longer that all codewords of given type are meant with respect to
$W$.\bigskip

\noindent In $Z^{n}$ there are $2^{2}\binom{n}{2}$ words $V$ of type $\left[
1^{2}\right]  $. Each of them is covered either by a codeword of type $\left[
1^{3}\right]  ,$ or by a codeword of type $\left[  2^{1},1^{1}\right]  .$
Further, each codeword of type $\left[  1^{3}\right]  $ covers three of them
while a codeword of type $\left[  2^{1},1^{1}\right]  $ covers exactly one
codeword of type $\left[  1^{2}\right]  .$ Hence%

\begin{equation}
\left\vert \lbrack 2^{1},1^{1}]\right\vert +3\left\vert \left[ 1^{3}\right]
\right\vert =4\binom{n}{2}  \label{pp}
\end{equation}

\noindent Equation (\ref{p}) and (\ref{pp}) are "global" equations. To get
their "local" form we need to introduce some more notation. Often we will need
to express the number of words, or codewords, in a set $\mathcal{A}$ having
their $i$-th coordinate positive, or their $i$-th coordinate negative.
Therefore, to simplify the language, we will introduce the notion of the
$signed$ $coordinate$ in $Z^{n}.$ For the rest of the paper by the set of
signed coordinates we will understand the set $I=\{+1,...,+n,-1,...,-n\}.$ Let
$V=(v_{1},...,v_{n})$ be a word in $Z^{n}.$ Then the signed coordinates
$V_{i}$ of $V$ are given by: $V_{i}=\left\vert v_{i}\right\vert $ and
$V_{-i}=0$ for $v_{i}>0,$ $V_{i}=0$ and $V_{-i}=\left\vert v_{i}\right\vert $
for $v_{i}<0,$ and $V_{i}=V_{-i}=0$ for $v_{i}=0.$ E.g., if $V=(2,0,-5)$ then
$V_{1}=2,V_{-1}=0,V_{2}=V_{-2}=0,$ and$V_{3}=0,$ $V_{-3}=5.$ For a signed
coordinate $i\in I,$ by $\left\vert \mathcal{A}_{i}\right\vert $ we will
denote the number of words in $\mathcal{A}$ with a non-zero $i$-th coordinate.
That is, $\left\vert A_{1}\right\vert $ stands for the number of words in
$\mathcal{A}$ with the first coordinate being a positive number, while
$\left\vert \mathcal{A}_{-3}\right\vert $ represents the number of words in
$\mathcal{A}$ with the third coordinate being a negative number. If we need to
stress that the value of the $i$-th signed coordinate is $m,$ we will use
$\left\vert \mathcal{A}_{i}^{(m)}\right\vert $ for the number of words with
the $i$-th coordinate equal $m.$ Thus, for each $i\in I,$ $\left\vert \left[
2^{1},1^{1}\right]  _{i}\right\vert $ is the number of words of type $\left[
2^{1},1^{1}\right]  $ with the $i$-th signed coordinate being non-zero, while
$\left\vert \left[  2^{1},1^{1}\right]  _{i}^{(2)}\right\vert $ stands for the
set of codewords of type $\left[  2^{1},1^{1}\right]  $ with the $i$-th signed
coordinate equal to $2.$ \bigskip

\noindent Now we are ready to state the local form of (\ref{p}) and (\ref{pp}%
). As for each $i\in I$ there is in $Z^{n}$ one word $V$ of type $[2^{1}]$
with $V_{i}=2,$ and $2(n-1)$ words $U$ of type $[1^{2}]$ with $U_{i}=1$, we
get:

\begin{equation}
\left\vert \lbrack 3^{1}]_{i}\right\vert +\left\vert
[2^{1},1^{1}]_{i}^{(2)}\right\vert =1,  \label{p1}
\end{equation}

\noindent and%
\begin{equation}
\left\vert \lbrack 2^{1},1^{1}]_{i}\right\vert +2\left\vert
[1^{3}]_{i}\right\vert =2(n-1).  \label{pp1}
\end{equation}

\noindent Indeed, if $A$ is a codeword of type $[3^{1}]$ with $A_{i}=3$ (and
then $A_{j}=0$ for all $j\neq i,j\in I)$ then $A$ covers a word $V$ of type $%
[2^{1}]$ with $V_{i}=2.$ However, a codeword $B$ of type $[2^{1},1^{1}]$
covers $V$ only if $B_{i}=2,$ but does not cover it if $B_{i}=1.$ On the
other hand, $\ $a codeword $B$ with $B_{i}\neq 0$ covers one word $D$ of
type $[1^{2}]$ with $D_{i}=1$ regardless whether $B_{i}=2$ or $B_{i}=1.$
Clearly, a codeword $C$ of type $[1^{3}]$ with $C_{i}=1$ covers exactly two
words $D$ of type $[1^{2}]$ with $D_{i}=1.$\bigskip

\noindent Now we derive identities analogous to (\ref{p}) - (\ref{pp1}) for
words of weight equal to $3.$ As (\ref{p}) - (\ref{pp1}) have been derived in
great detail, and the same type of ideas are used to prove identities
(\ref{a}) - (\ref{c1}) we will leave a part of the proofs to the
reader.\bigskip\ \

\noindent In $Z^{n}$ there are $2n$ words of type $[3^{1}].$ Each of them is
covered by a codeword of type $[3^{1}]$ or $[4^{1}]$ or $[3^{1},1^{1}],$ and
each of those codewords covers exactly one word of type $[3^{1}].$ Therefore,%
\begin{equation}
\left\vert \lbrack 3^{1}]\right\vert +\left\vert [4^{1}]\right\vert
+\left\vert [3^{1},1^{1}]\right\vert =2n,  \label{a}
\end{equation}%
and, for each $i\in I,$ we have%
\begin{equation}
\left\vert \lbrack 3^{1}]_{i}\right\vert +\left\vert [4^{1}]_{i}\right\vert
+\left\vert [3^{1},1^{1}]_{i}^{(3)}\right\vert =1.  \label{a1}
\end{equation}%
Further, in $Z^{n}$ there are $2^{3}\binom{n}{2}$ words of type $%
[2^{1},1^{1}].$ They are covered by codewords of type $[2^{1},1^{1}],$ or $%
[3^{1},1^{1}],$ or $\left[ 2^{2}\right] ,$ or $\left[ 2^{1},1^{2}\right] $.
Each codeword of type $\left[ 2^{2}\right] ,$ or $\left[ 2^{1},1^{2}\right] $
covers two such words, while each codeword of type $[2^{1},1^{1}]$, or $%
[3^{1},1^{1}]$ covers one of them. Hence%
\begin{equation}
\left\vert \lbrack 2^{1},1^{1}]\right\vert +\left\vert
[3^{1},1^{1}]\right\vert +2\left\vert [2^{2}]\right\vert +2\left\vert \left[
2^{1},1^{2}\right] \right\vert =2^{3}\binom{n}{2}  \label{b}
\end{equation}%
\bigskip The above identity has two local forms. There are $2(n-1)$ words $U$
of type $[2^{1},1^{1}]$ with $U_{i}=2,$ and $2(n-1)$ words $U$ of type $%
[2^{1},1^{1}]$ with $U_{i}=1.$ For each $i\in I$ we get%
\begin{equation}
\left\vert \left[ 2^{1},1^{1}\right] _{i}^{(2)}\right\vert +\left\vert
[3^{1},1^{1}]_{i}^{(3)}\right\vert +\left\vert [2^{2}]_{i}\right\vert
+2\left\vert \left[ 2^{1},1^{2}\right] _{i}^{(2)}\right\vert =2(n-1),
\label{b1}
\end{equation}%
and%
\begin{equation}
\left\vert \left[ 2^{1},1^{1}\right] _{i}^{(1)}\right\vert +\left\vert
[3^{1},1^{1}]_{i}^{(1)}\right\vert +\left\vert [2^{2}]_{i}\right\vert
+\left\vert \left[ 2^{1},1^{2}\right] _{i}^{(1)}\right\vert =2(n-1).
\label{b2}
\end{equation}%
Further, in $Z^{n}$ there are $2^{3}\binom{n}{3}$ words of type $[1^{3}].$
They are covered by codewords of type $[1^{3}$], or $\left[ 2^{1},1^{2}%
\right] $ , or $[1^{4}]$. Each codeword of type $[1^{4}]$ covers four of
them. Hence,%
\begin{equation}
\left\vert \left[ 1^{3}\right] \right\vert +\left\vert \left[ 2^{1},1^{2}%
\right] \right\vert +4\left\vert \left[ 1^{4}\right] \right\vert =2^{3}%
\binom{n}{3}  \label{c}
\end{equation}%
\bigskip The local form of (\ref{c}) reads as follows:%
\begin{equation}
\left\vert \left[ 1^{3}\right] _{i}\right\vert +\left\vert \left[ 2^{1},1^{2}%
\right] _{i}\right\vert +3\left\vert \left[ 1^{4}\right] _{i}\right\vert
=2^{2}\binom{n-1}{2}  \label{c1}
\end{equation}%
as in $Z^{n}$ there are $2^{2}\binom{n-1}{2}$ words $U$ of type $[1^{3}]$
with $U_{i}=1,$ and each codeword $V$ of type $[1^{4}]$ with $V_{i}=1$
covers three of them.\bigskip 

\noindent Clearly, there are many solutions of (\ref{p}),...,(\ref{c1}) in
natural numbers. We will prove, that only one corresponds to a tiling of
$R^{n}$ by crosses.\bigskip

\noindent We will split Theorem \ref{5} into two statements but will determine
also the local values for individual types. We start with the number of
codewords of weight $3$.

\begin{theorem}
\label{6}Let $n=2(\operatorname{mod}3\dot{)},$ $\mathcal{L}$ be a tiling of
$R^{n}$ by crosses, and $W$ be a codeword. Then, the number of codewords of
given type with respect to $W$ is: $\left\vert [3^{1}]\right\vert =0,$
$\left\vert [2^{1},1^{1}]\right\vert =2n,$ and $\left\vert \left[
1^{3}\right]  \right\vert =\frac{2n(n-2)}{3}.$ As to the local values, for
each $i\in I,$ $\left\vert \left[  2^{1},1^{1}\right]  _{i}^{(2)}\right\vert
=\left\vert \left[  2^{1},1^{1}\right]  _{i}^{(1)}\right\vert =1,$ that is,
$\left\vert \left[  2^{1},1^{1}\right]  _{i}\right\vert =2,$ and $\left\vert
\left[  1^{3}\right]  _{i}\right\vert =n-2.$
\end{theorem}

\begin{proof}
Let $W$ be a codeword in $\mathcal{T}_\mathcal{L}$. Clearly, then also the set
$\mathcal{T}^{\prime}=\{U,U\in Z^{n},U=V-W$ for some $V$ in $\mathcal{T}_\mathcal{L}\}$ is
a tiling of $Z^{n}$ by Lee spheres. Therefore, wlog we assume $W=O$. From
(\ref{p1}) we have $\left\vert \left[  2^{1},1^{1}\right]  _{i}^{(2)}%
\right\vert \leq1, $ while from (\ref{pp1}) we get $\left\vert \left[
2^{1},1^{1}\right]  _{i}\right\vert $ is even, hence $\left\vert \left[
2^{1},1^{1}\right]  _{i}^{(2)}\right\vert \leq\left\vert \left[  2^{1}%
,1^{1}\right]  _{i}\right\vert $. On the other hand, there is no $i\in I$ with
$\left\vert \left[  2^{1},1^{1}\right]  _{i}^{(2)}\right\vert <\left\vert
\left[  2^{1},1^{1}\right]  _{i}\right\vert $as $%
%TCIMACRO{\dsum \limits_{i\in I}}%
%BeginExpansion
{\displaystyle\sum\limits_{i\in I}}
%EndExpansion
\left\vert \left[  2^{1},1^{1}\right]  _{i}^{(2)}\right\vert =%
%TCIMACRO{\dsum \limits_{i\in I}}%
%BeginExpansion
{\displaystyle\sum\limits_{i\in I}}
%EndExpansion
\left\vert \left[  2^{1},1^{1}\right]  _{i}^{(1)}\right\vert .$ Thus we
proved:\bigskip

\noindent\textbf{Lemma A.} For each $i\,\in I,$ either $\left\vert \left[
2^{1},1^{1}\right]  _{i}\right\vert =0$ or $\left\vert \left[  2^{1}%
,1^{1}\right]  _{i}\right\vert =2.$ In the latter case $\left\vert \left[
2^{1},1^{1}\right]  _{i}^{(2)}\right\vert =\left\vert \left[  2^{1}%
,1^{1}\right]  _{i}^{(1)}\right\vert =1.$\bigskip

\noindent Now we are ready to prove that $\left\vert [3^{1}]\right\vert =0.$
We consider two cases.\bigskip

\noindent(i) Let $\left\vert [3^{1}]_{i}\right\vert =1.$ Then, by (\ref{a1}),
$\left\vert [3^{1},1^{1}]_{i}^{(3)}\right\vert =0,$ and by (\ref{p1}),
$\left\vert \left[  2^{1},1^{1}\right]  _{i}^{(2)}\right\vert =0,$ which
implies, by Lemma A, that $\left\vert \left[  2^{1},1^{1}\right]
_{i}\right\vert =0.$ This in turn implies, see (\ref{pp1}), $\left\vert
\left[  1^{3}\right]  _{i}\right\vert =n-1.$ Substituting it into (\ref{c1})
gives $\left\vert \left[  2^{1},1^{2}\right]  _{i}\right\vert +3\left\vert
\left[  1^{4}\right]  _{i}\right\vert =(n-1)(2n-5).$ As we deal with the case
$n=2(\operatorname{mod}3),$ then \mbox{$(n-1)(2n-5)$}$=2(\operatorname{mod}3)$ as well,
and therefore $\left\vert \left[  2^{1},1^{2}\right]  _{i}\right\vert
=2(\operatorname{mod}3).$ Subtracting (\ref{b2}) from (\ref{b1}), and using
$\left\vert [3^{1},1^{1}]_{i}^{(3)}\right\vert =\left\vert \left[  2^{1}%
,1^{1}\right]  _{i}^{(2)}\right\vert =\left\vert \left[  2^{1},1^{1}\right]
_{i}^{(1)}\right\vert =0,$ we get $2\left\vert \left[  2^{1},1^{2}\right]
_{i}^{(2)}\right\vert =\left\vert \left[  2^{1},1^{2}\right]  _{i}%
^{(1)}\right\vert +\left\vert [3^{1},1^{1}]_{i}^{(1)}\right\vert .$ As
$\left\vert \left[  2^{1},1^{2}\right]  _{i}\right\vert =\left\vert \left[
2^{1},1^{2}\right]  _{i}^{(2)}\right\vert +\left\vert \left[  2^{1}%
,1^{2}\right]  _{i}^{(1)}\right\vert ,$ adding $\left\vert \left[  2^{1}%
,1^{2}\right]  _{i}^{(2)}\right\vert $ to both sides yields $3\left\vert
\left[  2^{1},1^{2}\right]  _{i}^{(2)}\right\vert =\left\vert \left[
2^{1},1^{2}\right]  _{i}\right\vert $\newline$+ \left\vert [3^{1},1^{1}]_{i}%
^{(1)}\right\vert .$ We showed above that in this case of $\left\vert
[3^{1}]_{i}\right\vert =1$ it is $\left\vert \left[  2^{1},1^{2}\right]
_{i}\right\vert =2(\operatorname{mod}3).$ Therefore $\left\vert [3^{1}%
,1^{1}]_{i}^{(1)}\right\vert >0,$ that is, $\left\vert [3^{1},1^{1}]_{i}%
^{(1)}\right\vert >\left\vert [3^{1},1^{1}]_{i}^{(3)}\right\vert .$\bigskip

\noindent(ii) Now let $\left\vert [3^{1}]_{i}\right\vert =0.$ By (\ref{p1}), we get
$\left\vert \left[  2^{1},1^{1}\right]  _{i}^{(2)}\right\vert =1,$ which
implies, by \mbox{Lemma A}, that $\left\vert \left[  2^{1},1^{1}\right]
_{i}\right\vert =2.$ This in turn implies, see (\ref{pp1}), $\left\vert
\left[  1^{3}\right]  _{i}\right\vert =n-2.$ Substituting it into (\ref{c1})
gives $\left\vert \left[  2^{1},1^{2}\right]  _{i}\right\vert +3\left\vert
\left[  1^{4}\right]  _{i}\right\vert =(n-2)(2n-3).$ As
$n=2(\operatorname{mod}3),$ it is $(n-2)(2n-3)=0(\operatorname{mod}3),$ and
therefore $\left\vert \left[  2^{1},1^{2}\right]  _{i}\right\vert
=0(\operatorname{mod}3).$ Subtracting (\ref{b2}) from (\ref{b1}), and using
$\left\vert \left[  2^{1},1^{1}\right]  _{i}^{(2)}\right\vert =\left\vert
\left[  2^{1},1^{1}\right]  _{i}^{(1)}\right\vert \mbox{=1},$ we get $2\left\vert
\left[  2^{1},1^{2}\right]  _{i}^{(2)}\right\vert +\left\vert [3^{1}%
,1^{1}]_{i}^{(3)}\right\vert -\left\vert [3^{1},1^{1}]_{i}^{(1)}\right\vert
=\left\vert \left[  2^{1},1^{2}\right]  _{i}^{(1)}\right\vert \,,$ and adding
$\left\vert \left[  2^{1},1^{2}\right]  _{i}^{(2)}\right\vert $ to both sides
gives $3\left\vert \left[  2^{1},1^{2}\right]  _{i}^{(2)}\right\vert
+\left\vert [3^{1},1^{1}]_{i}^{(3)}\right\vert -\left\vert [3^{1},1^{1}%
]_{i}^{(1)}\right\vert =\left\vert \left[  2^{1},1^{2}\right]  _{i}\right\vert
.$ As $\left\vert \left[  2^{1},1^{2}\right]  _{i}\right\vert
=0(\operatorname{mod}3)$ in this case, we have $\left\vert [3^{1},1^{1}%
]_{i}^{(3)}\right\vert -\left\vert [3^{1},1^{1}]_{i}^{(1)}\right\vert
=0(\operatorname{mod}3),$ which yields $\left\vert [3^{1},1^{1}]_{i}%
^{(1)}\right\vert \geq\left\vert \lbrack3^{1},1^{1}]_{i}^{(3)}\right\vert $ as
\newline $\left\vert [3^{1},1^{1}]_{i}^{(3)}\right\vert 
\leq 1$ for all $i\in I,$ see (\ref{a1}).\bigskip

\noindent So, $\left\vert [3^{1}]_{i}\right\vert =1$ implies $\left\vert
[3^{1},1^{1}]_{i}^{(1)}\right\vert >\left\vert [3^{1},1^{1}]_{i}%
^{(3)}\right\vert ,$ while $\left\vert [3^{1}]_{i}\right\vert =0$ gives
$\left\vert [3^{1},1^{1}]_{i}^{(1)}\right\vert \geq\left\vert \lbrack
3^{1},1^{1}]_{i}^{(3)}\right\vert .$ However, $%
%TCIMACRO{\dsum \limits_{i\in I}}%
%BeginExpansion
{\displaystyle\sum\limits_{i\in I}}
%EndExpansion
\left\vert [3^{1},1^{1}]_{i}^{(1)}\right\vert =%
%TCIMACRO{\dsum \limits_{i\in I}}%
%BeginExpansion
{\displaystyle\sum\limits_{i\in I}}
%EndExpansion
\left\vert [3^{1},1^{1}]_{i}^{(3)}\right\vert ,$ therefore there is no $i\in
I$ with $\left\vert [3^{1}]_{i}\right\vert =1,$ that is $\left\vert
[3^{1}]\right\vert =0,$ and, for all $i\in I,$
\begin{equation}
\left\vert \lbrack3^{1},1^{1}]_{i}^{(1)}\right\vert =\left\vert [3^{1}%
,1^{1}]_{i}^{(3)}\right\vert ,\text{ and }3\left\vert \left[  2^{1}%
,1^{2}\right]  _{i}^{(2)}\right\vert =\left\vert \left[  2^{1},1^{2}\right]
_{i}\right\vert \label{s}%
\end{equation}
Since $\left\vert [3^{1}]\right\vert =0,$ by (\ref{p}) we get $\left\vert
[2^{1},1^{1}]\right\vert =2n,$ which in turn implies, by (\ref{pp}), that
$\left\vert \left[  1^{3}\right]  \right\vert =\frac{2n(n-2)}{3}.$ Further,
from $\left\vert \left[  2^{1},1^{1}\right]  _{i}\right\vert =2,$ we get
$\left\vert \left[  1^{3}\right]  _{i}\right\vert =n-2.$ The proof is complete.
\end{proof} \bigskip

\noindent Now we prove an analogue of Theorem \ref{6} for the values of $\left\vert \left[  2^{1},1^{2}\right]\right\vert$
and $\left\vert \left[  1^{4}\right]\right\vert.$

\begin{theorem}
\label{7} Let $\mathcal{L}$ be a tiling of $R^{n}$ by crosses where
$n=2(\operatorname{mod}3\dot{)}.$ Then, for each $W\in\mathcal{T}_\mathcal{L}$,
$\left\vert \left[ 2^{1}, 1^{2} \right] \right\vert =2n(n-2),$ and $\left\vert \left[  1^{4}\right]\right\vert=\frac{n(n-2)(n-3)}{3}.\,$\ In addition, for all $i\in I,$
it is$,\left\vert \left[  2^{1},1^{2}\right] _{i}\right\vert=3(n-2),$ $\left\vert \left[  2^{1},1^{2}\right] _{i}^{(2)}\right\vert=n-2,$ and $\left\vert \left[  1^{4}\right] _{i}\right\vert=\frac{2(n-2)(n-3)}{3}.$
\end{theorem}

\begin{proof}
As with Theorem \ref{6}, w.l.o.g we assume that $W=O.$ In order to determine
the value of $\left\vert \left[  2^{1},1^{2}\right]\right\vert$ we need the following lemma:

\begin{lemma}
\label{L2} For each $i\in I,$ it is $\left\vert \left[  2^{2} \right] _{i}\right\vert\leq1;$ hence $\left\vert \left[  2^{2} \right]\right\vert\leq n/2.$
\end{lemma}

\noindent\textbf{Proof of Lemma }\ref{L2}. Assume by contradiction that there
is $i\in I,\,\ $say $i=1,\,\ $such that $\left\vert \left[  2^{2} \right] _{1}\right\vert\geq2.$ Let,~w.l.o.g,
$F=(2,2,0,...,0),F^{\prime}=(2,0,2,0,...,0)$ be two codewords of type
$[2^{2}]$ with $F_{1}=F_{1}^{\prime}=2$. We proved that, for each $i\in I,$
it is $\left\vert \left[  2^1, 1^1 \right] _{i}^{(2)}\right\vert=1.$ So there is a codeword $B$ of type $[2^1, 1^1]$, with
$B_{1}=2.$ We may assume w.l.o.g. that $B=(2,...,0,\pm1,0,...0).$ If
$B=(2,-1,0,...,0),$ then $F_{1}-B=(0,3,0,...,0),$ that is, the codeword $F_{1}
$ is with respect to the codeword $B$ of type $[3^1],$ which is a
contradiction as we proved that $\left\vert \left[  3^{1} \right] \right\vert=0.$ So let $B=(2,0,0,1,0,...,0).$ Then
$F_{1}-B=(0,2,0,-1,0,...,0)$ and $F_{2}-B=(0,0,2,-1,0,...,0).$ That is, with
respect to the codeword $B,$ we get $\left\vert \left[ 2^1, 1^1 \right]_{i}^{(1)} \right\vert=2,$ which contradicts that
$\left\vert \left[ 2^1, 1^1 \right]_i^{(1)} \right\vert=1$ for all $i\in I.$ Therefore, $\left\vert \left[ 2^2 \right] _{i} \right\vert\leq1$ for all $i\in I,$
which in turn implies $\left\vert \left[ 2^2 \right] \right\vert=\frac{1}{2}%
%TCIMACRO{\dsum \limits_{i\in I}}%
%BeginExpansion
{\displaystyle\sum\limits_{i\in I}}
%EndExpansion
\left\vert \left[ 2^2\right] _i \right\vert\leq n.$ This proves Lemma \ref{L2}.\bigskip

\noindent With this in hand we find the values of $\left\vert \left[ 2^{1}, 1^{2} \right] \right\vert$ and $\left\vert \left[ 2^{1}, 1^{2} \right]_{i} \right\vert.$ The equality (\ref{b1}) states that $\left\vert \left[ 2^{1}, 1^{1} \right]_i^{(2)} \right\vert+\left\vert \left[ 3^{1}, 1^{1} \right]_i^{(3)} \right\vert+\left\vert \left[ 2^{2} \right]_i \right\vert+2\left\vert \left[ 2^{1}, 1^{2} \right]_{i}^{(2)} \right\vert=$\mbox{$2(n-1)$}. In addition, by Lemma \ref{6}, it is $\left\vert \left[ 2^1, 1^1 \right]_i^{(2)} \right\vert=1,$ by
(\ref{a1}) $\left\vert \left[ 3^1, 1^1 \right]_i^{(3)} \right\vert\leq1,$ and by the above lemma $\left\vert \left[ 2^2 \right]_i \right\vert\leq1.$ As
$\left\vert \left[ 2^1, 1^1 \right]_i^{(2)} \right\vert+\left\vert \left[ 3^1, 1^1 \right]_i^{(3)} \right\vert+\left\vert \left[ 2^2 \right]_i \right\vert$ is an even number, we get
\begin{equation}
\left\vert \left[ 2^1, 1^1 \right]_i^{(2)} \right\vert+\left\vert \left[ 3^1, 1^1 \right]_i^{(3)} \right\vert+\left\vert \left[ 2^2 \right]_i \right\vert=2.\label{ee}%
\end{equation}
Therefore $\left\vert \left[ 2^{1}, 1^{2} \right]_{i}^{(2)} \right\vert=n-2,$ thus $\left\vert \left[ 2^{1}, 1^{2} \right] \right\vert=$ $%
%TCIMACRO{\dsum \limits_{i\in I}}%
%BeginExpansion
{\displaystyle\sum\limits_{i\in I}}
%EndExpansion
\left\vert \left[ 2^{1}, 1^{2} \right]_{i}^{(2)} \right\vert=2n(n-2).$ By (\ref{s}), $\left\vert \left[ 2^{1}, 1^{2} \right]_{i} \right\vert=3\left\vert \left[ 2^{1}, 1^{2} \right]_{i}^{(2)} \right\vert=3(n-2).$ The values of
$\left\vert \left[ 1^{4} \right] \right\vert$ and $\left\vert \left[ 1^{4} \right]_{i} \right\vert$ are easily obtained from (\ref{c}) and (\ref{c1}),
respectively. The proof is complete.\bigskip
\end{proof}

\noindent To be able to determine the values of $\left\vert \left[ 4^{1} \right] \right\vert,\left\vert \left[ 3^{1}, 1^{1} \right] \right\vert,$ and $\left\vert \left[ 2^{2} \right] \right\vert,$ we need to
consider codewords from the $5$-neighbourhood. Each word $V,\left\vert
V\right\vert _{M}=5,$ is either of type $\left[  5^1\right]  ,$ or $\left[
4^1,1^1\right]  ,$ or $\left[  3^1,2^1\right]  ,$ or $[3^1,1^{2}], $
or $\left[  2^{2},1^1\right]  ,$ or $\left[  2^1,1^{3}\right]  ,$ or
$\left[  1^{5}\right]  .$ Let $W$ be a codeword in $\mathcal{T}_\mathcal{L}$. Then the
number of codewords $Z$ in the $5$-neighbourhood of $W$ of the given type
$\left[  5^1\right]  $ will be denoted by $\left\vert \left[ 5^1 \right] \right\vert,$ of type $\left[  4^1,1^1\right]  $ by $\left\vert \left[ 4^1, 1^{1} \right] \right\vert,$ etc.\bigskip

\noindent We start with a series of auxiliary statements.

\begin{lemma}
\label{L3}For each $i\in I,$ $\left\vert \left[ 3^{1}, 2^{1} \right]_{i}^{(3)} \right\vert\leq1,~$and $\left\vert \left[ 3^{1}, 2^{1} \right]_{i}^{(2)} \right\vert\leq2.$
\end{lemma}

\begin{proof}
Assume by contradiction that there are two codewords $C^{1}$ and $C^{2}$ of
type $[3^1,2^1]$ with $C_{i}^{k}=3$ for $k=1,2.$ By (\ref{a1}) we have
$\left\vert \left[ 4^1 \right]_i \right\vert+\left\vert \left[ 3^1, 1^1 \right]_i^{(3)} \right\vert=1$ as we know from Theorem \ref{6} that $\left\vert \left[3^1 \right]_i \right\vert=0.$\bigskip

\noindent So, assume first that $\left\vert \left[ 4^{1} \right]_{i} \right\vert=1.$ Then there is a codeword $D$, with
$D_{i}=4.$ Say, w.l.o.g, $D=(4,0,...,0),$ and $C^{1}=(3,2,0,...,0),$
$C^{2}=(3,0,2,0,...,0)$. As $C^{1}-D=(-1,2,0,...,0),$ and $C^{2}%
-D=(-1,0,2,0,...,0),$ we arrived at a contradiction since with respect to $D$
we have $\left\vert \left[ 2^{1}, 1^{1} \right]_{i}^{(1)} \right\vert>1.$\bigskip

\noindent Suppose now that $\left\vert \left[ 3^1, 1^1 \right]_i^{(3)} \right\vert=1,$ i.e., there is codeword $E$ of
type $[3^{1},1^{1}]$ so that $E_{i}=3;$ say $E=(3,1,0,...,0\dot{)}.$ If
$C^{1}=(3,-2,0,...,0)$ then $C^{1}-E=(0,3,0,...,0)\,\ $a contradiction as $\left\vert \left[ 3^{1} \right] \right\vert=0
$ with respect to all codewords. So we may assume that $C^{1}=(3,0,2,0,..,0),
$ and $C^{2}=(3,0,0,2,0,...,0).$ Then $C^{1}-E=(0,-1,2,0,...,0),$ and
$C^{2}-E=(0,-1,0,2,0,...,0);$ i.e., with respect to $W,$ $\left\vert \left[ 2^{1}, 1^{1} \right] _{i}^{(1)} \right\vert>1,~$\ a
contradiction. The proof of the first part follows.\bigskip

\noindent Now let $B$ be a codeword of type $\left[  2^{1},1^{1}\right]  $ with
$B_{i}=2.$ Further, let $\left\vert \left[ 3^{1}, 2^{1} \right]_{i}^{(2)} \right\vert\newline\geq3$, and $C^{1},C^{2},$ and
$C^{3}$ be codewords of type $[3^{1},2^{1}]$ with $C_{i}^{j}=2\,,$ $j=1,...,3.$
We assume w.l.o.g. that $i=1,$ and $B=(2,1,0,..,0).$ Then there are at least
two of the codewords $C^{j},$ say $C^{1}$ and $C^{2}$ having the second
coordinate equal to $0,$ as otherwise the two codewords would be at distance
less than $3$. We assume w.l.o.g. $C^{1}=(2,0,3,0,...,0),$ and $C^{2}%
=(2,0,0,3,0,...,0).$ Hence $C^{1}-B=(0,-1,3,0,...,0)\,\ $\ and $C^{2}%
-B=(0,-1,0,3,0,...,0);$ i.e., with respect to the codeword $B\,\ $\ we get
$\left\vert \left[ 3^{1}, 1^{1} \right] _{-2}^{(1)} \right\vert>1.$ This contradicts (\ref{s}) because $\left\vert \left[ 3^1, 1^1 \right]_i^{(3)} \right\vert\leq1$ for
all $i\in I.$ The proof is complete.\bigskip
\end{proof}

\noindent Before we prove the next lemma we get equalities related to covering
words of absolute value $4.$ In $Z^{n}$ there are $2n$ words of type $[4^{1}].$
By Theorem \ref{6}, there is no codeword of type $[3^{1}].$ Hence we have%
\begin{equation}
\left\vert \left[ 4^{1} \right] \right\vert+\left\vert \left[ 5^{1} \right] \right\vert+\left\vert \left[ 4^{1}, 1^{1} \right] \right\vert=2n,\label{d}%
\end{equation}

\noindent and, for all $i\in I,$ the local form reads as follows:%
\begin{equation}
\left\vert \left[ 4^{1} \right]_{i} \right\vert+\left\vert \left[ 5^{1} \right] _{i} \right\vert+\left\vert \left[ 4^{1}, 1^{1} \right] _{i}^{(4)} \right\vert=1\label{d1}%
\end{equation}

\noindent Further, in $Z^{n}$ there are $2^{3}\binom{n}{2}$ words of type
$[3^{1},1^{1}].$ Each of them is covered by a codeword of type $\left[
3^{1}\right]  ,$ or $\left[  2^{1},1^{1}\right]  ,$ or $[ 3^{1}, 1^{1}]
,$ or $[ 4^{1}, 1^{1}]  ,$ or $[ 3^{1}, 2^{1}]  ,$ or
$\left[  3^{1},1^{2}\right]  .$ Only codewords of type $\left[  
3^{1},1^{2}\right]  $ cover two words of type $[3^{1},1^{1}].$ In addition
we know that there is no codeword of type $\left[  3^{1}\right]  $. Thus,%

\begin{equation}
\left\vert \left[ 2^{1}, 1^{1} \right] \right\vert+\left\vert \left[ 3^{1}, 1^{1} \right] \right\vert+\left\vert \left[ 4^{1}, 1^{1} \right] \right\vert+\left\vert \left[ 3^{1}, 2^{1} \right] \right\vert+2\left\vert \left[ 3^{1}, 1^{2} \right] \right\vert=8\binom{n}{2}\label{e}%
\end{equation}
For each $i\in I,$ there are $2(n-1)$ words $V$ of type $[3^{1},1^{1}]$ with
$V_{i}=3,$ and, at the same time, $2(n-1)$ words $V$ of type $[3^{1},1^{1}]$
with $V_{i}=1.$ Thus the local forms of (\ref{e}) read as follows:%
\begin{equation}
\begin{gathered}
\left\vert \left[ 2^1, 1^1 \right]_i^{(2)} \right\vert+\left\vert \left[ 3^1, 1^1 \right]_i^{(3)} \right\vert+\left\vert \left[ 4^{1}, 1^{1} \right] _{i}^{(4)} \right\vert+\left\vert \left[ 3^{1}, 2^{1} \right]_{i}^{(3)} \right\vert+\\2\left\vert \left[ 3^{1}, 1^{2} \right] _{i}^{(3)} \right\vert=2(n-1),\label{e1}%
\end{gathered}
\end{equation}

\noindent and
\begin{equation}
\begin{gathered}
\left\vert \left[ 2^{1}, 1^{1} \right] _{i}^{(1)} \right\vert+\left\vert \left[ 3^{1}, 1^{1} \right] _{i}^{(1)} \right\vert+\left\vert \left[ 4^{1}, 1^{1} \right] _{i}^{(1)} \right\vert+\left\vert \left[ 3^{1}, 2^{1} \right]_{i}^{(2)} \right\vert+\\ \left\vert \left[ 3^{1}, 1^{2} \right] _{i}^{(1)} \right\vert=2(n-1).\label{E2}%
\end{gathered}
\end{equation}

\noindent In $Z^{n}$ there are $2^{2}\binom{n}{2}$ words of type $[2^{2}].$
Each word of this type is covered by a codeword of type $\left[  
2^{1},1^{1}\right]  ,$ or $\left[  2^{2}\right]  ,$ or $\left[  
3^{1},2^{1}\right]  ,$ or $\left[  2^{2},1^{1}\right]  $. As each such codeword
covers one word of type $[2^{2}]$ we have
\begin{equation}
\left\vert \left[ 2^{1}, 1^{1} \right] \right\vert+\left\vert \left[ 2^{2} \right] \right\vert+\left\vert \left[ 3^{1}, 2^{1} \right] \right\vert+\left\vert \left[ 2^{2}, 1^{1} \right] \right\vert=2^{2}\binom{n}{2}\label{f}%
\end{equation}

\noindent and, for each $i\in I,$ we get%
\begin{equation}
\left\vert \left[ 2^{1}, 1^{1} \right] _{i} \right\vert+\left\vert \left[ 2^2 \right]_i \right\vert+\left\vert \left[ 3^{1}, 2^{1} \right] _{i} \right\vert+\left\vert \left[ 2^{2}, 1^{1} \right] _{i}^{(2)} \right\vert=2(n-1)\label{f1}%
\end{equation}

\noindent In $Z^{n}$ there are $3\cdot2^{3}\binom{n}{3}$ words of type
$[2^{1}, 1^{2}].$ Each of these words is covered by a codeword of type
$\left[  2^{1},1^{1}\right]  ,$ or $[1^{3}]  ,$ or $\left[
2^{1},1^{2}\right]  ,$ or $\left[  3^{1},1^{2}\right]  ,$ or $\left[
2^{1},1^{3}\right]  ,$ or $\left[  2^{2},1^{1}\right]  .$ As each codeword
of type $[2^{1}, 1^{1}]$ covers $2(n-2)$ of them, each codeword of type
$[1^{3}]$ and of type $[2^{1},1^{3}]$ covers three of them, and each
codeword of type $[2^{2}, 1^{1}]$ covers two of them, we get%
\begin{equation}
\begin{gathered}
2(n-2)\left\vert \left[ 2^{1}, 1^{1} \right] \right\vert+3\left\vert \left[ 1^{3} \right] \right\vert+\left\vert \left[ 2^{1}, 1^{2} \right] \right\vert+\left\vert \left[ 3^{1}, 1^{2} \right] \right\vert+\\ 3\left\vert \left[ 2^{1}, 1^{3} \right] \right\vert+2\left\vert \left[ 2^{2}, 1^{1} \right] \right\vert=24\binom{n}{3}\label{g}%
\end{gathered}
\end{equation}

\noindent For each $i\in I,$ in $Z^{n}$ there are $2^{2}\binom{n-1}{2}$ words
$V$ of type $[2^{1}, 1^{2}]$ with $V_{i}=2.$ Hence, a local form of (\ref{g})
is%
\begin{equation}
\begin{gathered}
2(n-2)\left\vert \left[ 2^1, 1^1 \right]_i^{(2)} \right\vert+\left\vert \left[  1^{3} \right] _{i} \right\vert+\left\vert \left[ 2^{1}, 1^{2} \right]_{i}^{(2)} \right\vert+\left\vert \left[ 3^{1}, 1^{2} \right] _{i}^{(3)} \right\vert+ \\ 3\left\vert \left[ 2^{1}, 1^{3} \right] _{i}^{(2)} \right\vert+\left\vert \left[ 2^{2}, 1^{1} \right] _{i}^{(2)} \right\vert=2^{2}\binom{n-1}{2}.\label{g1}%
\end{gathered}
\end{equation}
In $Z^{n}$ there are $8\binom{n-1}{2}$ words $V$ of type $\left[  
2^{1},1^{2}\right]  $ with $V_{i}=1.$ It is not difficult to see that:%
\begin{equation}
\begin{gathered}
\left\vert \left[ 2^{1}, 1^{1} \right] _{i}^{(0)} \right\vert+2(n-2)\left\vert \left[ 2^{1}, 1^{1} \right] _{i}^{(1)} \right\vert+2\left\vert \left[  1^{3} \right] _{i} \right\vert+\left\vert \left[ 2^{1}, 1^{2} \right] _{i}^{(1)} \right\vert+\\ \left\vert \left[ 3^{1}, 1^{2} \right] _{i}^{(1)} \right\vert+2\left\vert \left[ 2^{1}, 1^{3} \right] _{i}^{(1)} \right\vert+\left\vert \left[ 2^{2}, 1^{1} \right] _{i}^{(2)} \right\vert+\\ 2\left\vert \left[ 2^{2}, 1^{1} \right] _{i}^{(1)} \right\vert=8\binom{n-1}%
{2},\label{g2}%
\end{gathered}
\end{equation}

\noindent where $\left\vert \left[ 2^{1}, 1^{1} \right] _{i}^{(0)} \right\vert$stands for codeword $V$ of type $\left[ 2^{1},1^{1}\right]  $ with $V_{i}=V_{-i}=0.$ As $\left\vert \left[ 2^{1}, 1^{1} \right] _{i} \right\vert=\left\vert \left[ 2^{1}, 1^{1} \right] _{-i} \right\vert=2,$ we have
$\left\vert \left[ 2^{1}, 1^{1} \right] _{i}^{(0)} \right\vert=2(n-2).$ Finally, in $Z^{n}$ there are $2^{4}\binom{n}{4}$ words
of type $\left[ 1^{4}\right]  .$ Each of them is covered by a codeword of
type $[1^{3}]  ,$ or $\left[  1^{4}\right]  ,$ or $\left[
1^{5}\right]  ,$ or $\left[  2^{1},1^{3}\right]  .$ Clearly each codeword
of type $[1^{3}]  $ covers $2(n-3)$ words of type $\left[
1^{4}\right]  ,$ while each codeword of type $\left[  1^{5}\right]  $
covers $5$ words of type $\left[  1^{4}\right]  .$ Therefore:%
\begin{equation}
2(n-3)\left\vert \left[ 1^{3} \right] \right\vert+\left\vert \left[ 1^{4} \right]  \right\vert + \left\vert \left[ 2^{1}, 1^{3} \right] \right\vert+5\left\vert \left[ 1^{5} \right]  \right\vert=2^{4}\binom{n}{4}.\label{h}%
\end{equation}

\noindent The local form of (\ref{h}) reads as follows%
\begin{equation}
\begin{gathered}
\left\vert \left[ 1^{3} \right] _{i}^{(0)}  \right\vert+2(n-3)\left\vert \left[  1^{3} \right] _{i} \right\vert+\left\vert \left[ 1^{4} \right] _{i}  \right\vert+\left\vert \left[ 2^{1}, 1^{3} \right] _{i}  \right\vert+\\ 4\left\vert \left[ 1^{5} \right] _{i}  \right\vert=2^{3}\binom{n-1}%
{3},\label{h1}%
\end{gathered}
\end{equation}
where $\left\vert \left[ 1^{3} \right] _{i}^{(0)}  \right\vert$ stands for the number of codewords $V$ of type $\left[1^{3}\right]  $ with $V_{i}=V_{-i}=0.$\bigskip

\noindent Now we are ready to prove a lemma crucial for the proof of the next
theorem. (\ref{d1}) states that $\left\vert \left[ 4^{1} \right]_{i} \right\vert+\left\vert \left[ 5^{1} \right] _{i} \right\vert+\left\vert \left[ 4^{1}, 1^{1} \right] _{i}^{(4)} \right\vert=1.\,\ $Thus,
the lemma covers all possible cases.

\begin{lemma}
\label{L5}Let $i\in I.$ If $\left\vert \left[ 4^{1} \right]_{i} \right\vert=1,$ then $\left\vert \left[ 2^2 \right]_i \right\vert=1,$ $\left\vert \left[ 3^1, 1^1 \right]_i^{(3)} \right\vert=0,$ and
$\left\vert \left[ 3^{1}, 2^{1} \right]_{i}^{(3)} \right\vert=1,$ while if $\left\vert \left[ 5^{1} \right] _{i} \right\vert=1,$ then $\left\vert \left[ 2^2 \right]_i \right\vert=0,\left\vert \left[ 3^1, 1^1 \right]_i^{(3)} \right\vert=1,$
and $\left\vert \left[ 3^{1}, 2^{1} \right]_{i}^{(3)} \right\vert=0,$ but if $\left\vert \left[ 4^{1}, 1^{1} \right] _{i}^{(4)} \right\vert=1$ then $\left\vert \left[ 2^2 \right]_i \right\vert=0,\left\vert \left[ 3^{1}, 1^{1} \right] _{i}^{(3)}  \right\vert=1,$ and $\left\vert \left[ 3^{1}, 2^{1} \right]_{i}^{(3)} \right\vert=1.$ In particular, $\left\vert \left[ 3^{1}, 2^{1} \right] \right\vert=\left\vert \left[ 4^{1} \right] \right\vert+\left\vert \left[ 4^{1}, 1^{1} \right] \right\vert.$
\end{lemma}

\begin{proof}
Assume first that $\left\vert \left[ 4^{1} \right]_{i} \right\vert=1,$ and let $D$ be a codeword in $\mathcal{T}_\mathcal{L}$ of
type $\left[  4^{1}\right]  $ with $D_{i}=4$. As $\left\vert \left[ 4^{1} \right]_{i} \right\vert=1,$ then taking into
account $\left\vert \left[ 3^{1} \right] _{i} \right\vert=0$ and (\ref{a1}), we get $\left\vert \left[ 3^1, 1^1 \right]_i^{(3)} \right\vert=0.$ This in turn
implies, due to (\ref{ee}), that $\left\vert \left[ 2^2 \right]_i \right\vert=1.$ Consider the $3$-neighbourhood
$\mathcal{N}$ of $D$. \ By Theorem \ref{6}, we have for $\mathcal{N}$ that
$\left\vert \left[ 2^{1}, 1^{1} \right] _{i}^{(1)} \right\vert=1$ for all $i\in I.$ That is there has to be in $\mathcal{T}_\mathcal{L}$ a
codeword $B$ of type $[2^{1}, 1^{1}]$ with respect to $D,$ with $B_{i}-D_{i}=-1,$
that is $(B-D)_{-i}^{(1)}=1.$ Thus $B_{i}=3$ and there is a $j\in J$ so that
$B_{j}=2.$ Hence, $B$ is a codeword of type $[3^{1},2^{1}]$ with respect to the
origin $O\,,$ and therefore $\left\vert \left[ 3^{1}, 2^{1} \right]_{i}^{(3)} \right\vert\geq1,$ while by Lemma \ref{L3}
we get $\left\vert \left[ 3^{1}, 2^{1} \right]_{i}^{(3)} \right\vert=1.$ The first part of the proof is complete.\bigskip

\noindent Let now $\left\vert \left[ 5^{1} \right] _{i} \right\vert=1.$ Then there is a codeword $W$ in
$\mathcal{T}_\mathcal{L}$ of type $[5^{1}]$ with $W_{i}=5.$ Further, by (\ref{d1}),
$\left\vert \left[ 4^{1} \right]_{i} \right\vert=0,$ which in turn implies, see (\ref{a1}), that $\left\vert \left[ 3^1, 1^1 \right]_i^{(3)} \right\vert=1,$ and by
(\ref{ee}), $\left\vert \left[ 2^2 \right]_i \right\vert=0.$ Now we prove that in this case $\left\vert \left[ 3^{1}, 2^{1} \right]_{i}^{(3)} \right\vert=0$ as
well. Let $B\mathcal{\ }$be a codeword of type $\left[  2^{1},1^{1}\right]  $
with $B_{i}=2;$ w.l.o.g., let $W=(5,0,...,0),$ and $B=(2,1,0,...,0).$ Assume
that there is a codeword $C$ of type $[3^{1},2^{1}]$ with $C_{i}=3.$ Then
$B-W=(-3,1,0,...,0)$ and $C-W=(-2,0,2,0,...,0).$ That is, with respect to the
codeword $W,$ we have $\left\vert \left[ 3^{1}, 1^{1} \right] _{-1}^{(3)} \right\vert=1$ but also $\left\vert \left[ 2^{2} \right] _{-1} \right\vert=1,$ which contradicts
(\ref{ee}) as $\left\vert \left[ 2^1, 1^1 \right]_i^{(2)} \right\vert=1$ for all $i\in I.$ Thus $\left\vert \left[ 3^{1}, 2^{1} \right]_{i}^{(3)} \right\vert=0$ in
this case. The proof of the second part of the statement is complete.\bigskip

\noindent Now, assume that $\left\vert \left[ 4^{1}, 1^{1} \right] _{i}^{(4)} \right\vert=1.$ We need to show that in this
case $\left\vert \left[ 3^{1}, 2^{1} \right]_{i}^{(3)} \right\vert=1$ as well. However, by (\ref{e1}), $\left\vert \left[ 2^1, 1^1 \right]_i^{(2)} \right\vert%
+\left\vert \left[ 3^1, 1^1 \right]_i^{(3)} \right\vert+\left\vert \left[ 4^{1}, 1^{1} \right] _{i}^{(4)} \right\vert+\left\vert \left[ 3^{1}, 2^{1} \right]_{i}^{(3)} \right\vert$ is an even number, and in this
case we have $\left\vert \left[ 2^1, 1^1 \right]_i^{(2)} \right\vert=\left\vert \left[ 3^1, 1^1 \right]_i^{(3)} \right\vert=\left\vert \left[ 4^{1}, 1^{1} \right] _{i}^{(4)} \right\vert=1. $ Hence $\left\vert \left[ 3^{1}, 2^{1} \right] _{i}^{(3)} \right\vert$ is odd, and, by Lemma \ref{L3}, $\left\vert \left[ 3^{1}, 2^{1} \right]_{i}^{(3)} \right\vert=1.$ 

\noindent Finally, $\left\vert \left[ 3^{1}, 2^{1} \right]   \right\vert=%
%TCIMACRO{\dsum \limits_{i\in I}}%
%BeginExpansion
{\displaystyle\sum\limits_{i\in I}}
%EndExpansion
\left\vert \left[ 3^{1}, 2^{1} \right]_{i}^{(3)} \right\vert=%
%TCIMACRO{\dsum \limits_{\left\vert \left[ 4^{1} \right]_{i} \right\vert=1}}%
%BeginExpansion
{\displaystyle\sum\limits_{\left\vert \left[ 4^{1} \right]_{i} \right\vert=1}}
%EndExpansion
\left\vert \left[ 3^{1}, 2^{1} \right]_{i}^{(3)} \right\vert +\\%
%TCIMACRO{\dsum \limits_{\left\vert \left[ 4^{1}, 1^{1} \right] _{i}^{(4)} \right\vert=1}}%
%BeginExpansion
{\displaystyle\sum\limits_{\left\vert \left[ 4^{1}, 1^{1} \right] _{i}^{(4)} \right\vert=1}}
%EndExpansion
\left\vert \left[ 3^{1}, 2^{1} \right]_{i}^{(3)} \right\vert +%
%TCIMACRO{\dsum \limits_{\left\vert \left[ 5^{1} \right] _{i} \right\vert=1}}%
%BeginExpansion
{\displaystyle\sum\limits_{\left\vert \left[ 5^{1} \right] _{i} \right\vert=1}}
%EndExpansion
\left\vert \left[ 3^{1}, 2^{1} \right]_{i}^{(3)} \right\vert=%
%TCIMACRO{\dsum \limits_{\left\vert \left[ 4^{1} \right]_{i} \right\vert=1}}%
%BeginExpansion
{\displaystyle\sum\limits_{\left\vert \left[ 4^{1} \right]_{i} \right\vert=1}}
%EndExpansion
\left\vert \left[ 3^{1}, 2^{1} \right]_{i}^{(3)} \right\vert+\\%
%TCIMACRO{\dsum \limits_{\left\vert \left[ 4^{1}, 1^{1} \right] _{i}^{(4)} \right\vert=1}}%
%BeginExpansion
{\displaystyle\sum\limits_{\left\vert \left[ 4^{1}, 1^{1} \right] _{i}^{(4)} \right\vert=1}}
%EndExpansion
\left\vert \left[ 3^{1}, 2^{1} \right]_{i}^{(3)} \right\vert=\left\vert \left[ 4^{1} \right] \right\vert + \left\vert \left[ 4^{1}, 1^{1} \right] \right\vert.$ The proof is complete.\bigskip
\end{proof}

\noindent We are ready to determine the remaining values of $\left\vert \left[ 4^{1} \right] \right\vert,\left\vert \left[ 3^{1}, 1^{1} \right] \right\vert,$ and $\left\vert \left[ 2^{2} \right] \right\vert.$

\begin{theorem}
\label{8}For each codeword $W$in $\mathcal{T}_\mathcal{L},$ we have $\left\vert \left[ 3^{1}, 1^{1} \right] \right\vert=2n,$ $\ $and
$\left\vert \left[ 4^{1} \right] \right\vert=\left\vert \left[ 2^{2} \right] \right\vert=0.$ In addition, for all $i\in I,$ $\left\vert \left[ 3^1, 1^1 \right]_i^{(3)} \right\vert=\left\vert \left[ 3^{1}, 1^{1} \right] _{i}^{(1)} \right\vert\mbox{=1}.$
\end{theorem}

\begin{proof}
As above, we assume w.l.o.g. that $W=O.$ We recall that, by (\ref{s}),
$\left\vert \left[ 3^1, 1^1 \right]_i^{(3)} \right\vert=\left\vert \left[ 3^{1}, 1^{1} \right] _{i}^{(1)} \right\vert$for all $i\in I$. Since $\left\vert \left[ 3^{1}, 2^{1} \right] _{i} \right\vert=\left\vert \left[ 3^{1}, 2^{1} \right] _{i}^{(3)} \right\vert+\left\vert \left[ 3^{1}, 2^{1} \right]_{i}^{(2)} \right\vert,$ we get from (\ref{f1})
\begin{equation}
\begin{gathered}
\left\vert \left[ 2^{1}, 1^{1} \right] _{i} \right\vert+\left\vert \left[  2^{2} \right] _{i} \right\vert+\left\vert \left[ 3^{1}, 2^{1} \right]_{i}^{(3)} \right\vert+\left\vert \left[ 3^{1}, 2^{1} \right]_{i}^{(2)} \right\vert+\left\vert \left[ 2^{2}, 1^{1} \right] _{i}^{(2)} \right\vert\\
=2(n-1)\label{rr}%
\end{gathered}
\end{equation}
Combining Lemma \ref{L5} with (\ref{e1}), (\ref{E2}), \ and (\ref{rr}) yields:

\noindent If $\left\vert \left[ 4^{1} \right]_{i} \right\vert=1,$ then
\begin{equation}
\begin{gathered}
\left\vert \left[ 3^{1}, 1^{2} \right] _{i}^{(3)} \right\vert=n-2,\\ \left\vert \left[ 3^{1}, 1^{2} \right] _{i}^{(1)} \right\vert=2n-3-\left\vert \left[ 4^{1}, 1^{1} \right] _{i}^{(1)} \right\vert-\left\vert \left[  3^{1}, 2^{1} \right] _{i}^{(2)} \right\vert,\\ \text{ and }\left\vert \left[ 2^{2}, 1^{1} \right] _{i}^{(2)} \right\vert=2n-6-\left\vert \left[ 3^{1}, 2^{1} \right]_{i}^{(2)} \right\vert\label{r1}%
\end{gathered}
\end{equation}

\noindent For $\left\vert \left[  5^{1} \right] _{1} \right\vert=1,$ it is
\begin{equation}
\begin{gathered}
\left\vert \left[ 3^{1}, 1^{2} \right] _{i}^{(3)} \right\vert=n-2,\\ \left\vert \left[ 3^{1}, 1^{2} \right] _{i}^{(1)} \right\vert=2n-4-\left\vert \left[ 4^{1}, 1^{1} \right] _{i}^{(1)} \right\vert-\left\vert \left[  3^{1}, 2^{1} \right] _{i}^{(2)} \right\vert=,\\ \text{ and }\left\vert \left[ 2^{2}, 1^{1} \right] _{i}^{(2)} \right\vert=2n-4-\left\vert \left[ 3^{1}, 2^{1} \right]_{i}^{(2)} \right\vert,\label{r3}%
\end{gathered}
\end{equation}

\noindent and for $\left\vert \left[  4^{1}, 1^{1} \right] _{i} \right\vert=1$ we get%
\begin{equation}
\begin{gathered}
\left\vert \left[ 3^{1}, 1^{2} \right] _{i}^{(3)} \right\vert=n-3, \\ \left\vert \left[ 3^{1}, 1^{2} \right] _{i}^{(1)} \right\vert=2n-4-\left\vert \left[ 4^{1}, 1^{1} \right] _{i}^{(1)} \right\vert-\left\vert \left[  3^{1}, 2^{1} \right] _{i}^{(2)} \right\vert=,\\ \text{ and }\left\vert \left[ 2^{2}, 1^{1} \right] _{i}^{(2)} \right\vert=2n-5-\left\vert \left[ 3^{1}, 2^{1} \right]_{i}^{(2)} \right\vert.\label{r2}%
\end{gathered}
\end{equation}

\noindent Summing (\ref{r1}), (\ref{r3}), and (\ref{r2}) yields:

$\left\vert \left[ 3^{1}, 1^{2} \right] \right\vert=%
%TCIMACRO{\dsum \limits_{i\in I}}%
%BeginExpansion
{\displaystyle\sum\limits_{i\in I}}
%EndExpansion
\left\vert \left[ 3^{1}, 1^{2} \right] _{i}^{(3)} \right\vert=%
%TCIMACRO{\dsum \limits_{\left\vert \left[ 4^{1} \right]_{i} \right\vert=1}}%
%BeginExpansion
{\displaystyle\sum\limits_{\left\vert \left[ 4^{1} \right]_{i} \right\vert=1}}
%EndExpansion
\left\vert \left[ 3^{1}, 1^{2} \right] _{i}^{(3)} \right\vert+%
%TCIMACRO{\dsum \limits_{\left\vert \left[ 4^{1}, 1^{1} \right] _{i}^{(4)} \right\vert=1}}%
%BeginExpansion
{\displaystyle\sum\limits_{\left\vert \left[ 4^{1}, 1^{1} \right] _{i}^{(4)} \right\vert=1}}
%EndExpansion
\left\vert \left[ 3^{1}, 1^{2} \right] _{i}^{(3)} \right\vert+%
%TCIMACRO{\dsum \limits_{\left\vert \left[ 5^{1} \right] _{i} \right\vert=1}}%
%BeginExpansion
{\displaystyle\sum\limits_{\left\vert \left[ 5^{1} \right] _{i} \right\vert=1}}
%EndExpansion
\left\vert \left[ 3^{1}, 1^{2} \right] _{i}^{(3)} \right\vert=$ $\left\vert \left[ 4^{1} \right] \right\vert(n-2)+\left\vert \left[ 4^{1}, 1^{1} \right] \right\vert(n-3)+\left\vert \left[ 5^{1} \right] \right\vert(n-2).$ Hence, using
$\left\vert \left[ 4^{1} \right] \right\vert+\left\vert \left[ 4^{1}, 1^{1} \right] \right\vert+\left\vert \left[ 5^{1} \right] \right\vert=2n,$ see (\ref{d}), we get
\begin{equation}
\left\vert \left[ 3^{1}, 1^{2} \right] \right\vert=2n(n-2)-\left\vert \left[ 4^{1}, 1^{1} \right] \right\vert;\label{eee1}%
\end{equation}

\noindent and \bigskip

$2\left\vert \left[ 2^{2}, 1^{1} \right] \right\vert=%
%TCIMACRO{\dsum \limits_{i\in I}}%
%BeginExpansion
{\displaystyle\sum\limits_{i\in I}}
%EndExpansion
\left\vert \left[ 2^{2}, 1^{1} \right] _{i}^{(2)} \right\vert=%
%TCIMACRO{\dsum \limits_{\left\vert \left[ 4^{1} \right]_{i} \right\vert=1}}%
%BeginExpansion
{\displaystyle\sum\limits_{\left\vert \left[ 4^{1} \right]_{i} \right\vert=1}}
%EndExpansion
\left\vert \left[ 2^{2}, 1^{1} \right] _{i}^{(2)} \right\vert +%
%TCIMACRO{\dsum \limits_{\left\vert \left[ 4^{1}, 1^{1} \right] _{i}^{(4)} \right\vert=1}}%
%BeginExpansion
{\displaystyle\sum\limits_{\left\vert \left[ 4^{1}, 1^{1} \right] _{i}^{(4)} \right\vert=1}}
%EndExpansion
\left\vert \left[ 2^{2}, 1^{1} \right] _{i}^{(2)} \right\vert\\+%
%TCIMACRO{\dsum \limits_{\left\vert \left[ 5^{1} \right] _{i} \right\vert=1}}%
%BeginExpansion
{\displaystyle\sum\limits_{\left\vert \left[ 5^{1} \right] _{i} \right\vert=1}}
%EndExpansion
\left\vert \left[ 2^{2}, 1^{1} \right] _{i}^{(2)} \right\vert=$ $\left\vert \left[ 4^{1} \right] \right\vert(2n-6)+\left\vert \left[ 4^{1}, 1^{1} \right] \right\vert(2n-5)+\left\vert \left[ 5^{1} \right] \right\vert(2n-4)-%
%TCIMACRO{\dsum \limits_{i\in I}}%
%BeginExpansion
{\displaystyle\sum\limits_{i\in I}}
%EndExpansion
\left\vert \left[ 3^{1}, 2^{1} \right]_{i}^{(2)} \right\vert=2n(2n-6)+\left\vert \left[ 4^{1}, 1^{1} \right] \right\vert+2\left\vert \left[ 5^{1} \right] \right\vert-\left\vert \left[ 3^{1},2^{1} \right] \right\vert$\ as, see \mbox{Lemma \ref{L5}},
$\left\vert \left[ 3^{1}, 2^{1} \right]   \right\vert=\left\vert \left[ 4^{1} \right] \right\vert+\left\vert \left[ 4^{1}, 1^{1} \right] \right\vert.$ Thus
\begin{equation}
\left\vert \left[ 2^{2}, 1^{1} \right] \right\vert=n(2n-6)+\left\vert \left[ 5^{1} \right] \right\vert-\frac{\left\vert \left[ 4^{1} \right] \right\vert}{2}.\label{eee2}%
\end{equation}

\noindent Substituting into (\ref{g}) for $\left\vert \left[  2^{1}, 1^{1} \right]  \right\vert,\left\vert \left[  1^{3} \right]  \right\vert,\left\vert \left[ 2^{1},1^{2} \right] \right\vert$ from Theorem \ref{7} and for
$\left\vert \left[ 3^{1}, 1^{2} \right] \right\vert$ and $2\left\vert \left[ 2^{2}, 1^{1} \right] \right\vert$ from above we get
\begin{equation*}
\begin{gathered}
4n(n-2)+2n(n-2)+2n(n-2)+2n(n-2)-\left\vert \left[ 4^{1}, 1^{1} \right] \right\vert+ \\ 3\left\vert \left[ 2^{1}, 1^{3} \right] \right\vert+2n(2n-7)+\left\vert \left[ 4^{1}, 1^{1} \right] \right\vert+3\left\vert \left[ 5^{1} \right] \right\vert
=24\binom{n}{3},
\end{gathered}
\end{equation*}

\noindent i.e.,%
\begin{equation}
3\left\vert \left[ 2^{1}, 1^{3} \right] \right\vert+3\left\vert \left[ 5^{1} \right] \right\vert=2n(n-3)(2n-7)\label{i}%
\end{equation}

\noindent Substituting for $\left\vert \left[ 2^1, 1^1 \right]_i^{(2)} \right\vert,\left\vert \left[  1^{3} \right] _{i}  \right\vert,\left\vert \left[ 2^{1}, 1^{2} \right]_{i}^{(2)} \right\vert$ from Theorem
\ref{7}, to (\ref{g1}) yields $\left\vert \left[ 3^{1}, 1^{2} \right] _{i}^{(3)} \right\vert+3\left\vert \left[  2^{1}, 1^{3} \right] _{i}^{(2)}  \right\vert%
+\left\vert \left[ 2^{2}, 1^{1} \right] _{i}^{(2)} \right\vert=4\binom{n-1}{2}-4(n-2)=2(n-2)(n-3).$ Applying
(\ref{r1}), (\ref{r2}), \ and (\ref{r3}) in turn implies, if $\left\vert \left[ 4^{1} \right]_{i} \right\vert=1$ or
$\left\vert \left[ 4^{1}, 1^{1} \right] _{i}^{(4)} \right\vert=1, $ then
\begin{equation}
3\left\vert \left[  2^{1}, 1^{3} \right] _{i}^{(2)}  \right\vert=2(n-2)(n-3)-3n+8+\left\vert \left[ 3^{1}, 2^{1} \right]_{i}^{(2)} \right\vert,\label{lambda}%
\end{equation}
and if $\left\vert \left[ 5^{1} \right] _{i} \right\vert=1,$ then
\begin{equation}
3\left\vert \left[  2^{1}, 1^{3} \right] _{i}^{(2)}  \right\vert=2(n-2)(n-3)-3n+6+\left\vert \left[ 3^{1}, 2^{1} \right]_{i}^{(2)} \right\vert.\label{lambda1}%
\end{equation}

\noindent As $n=2(\operatorname{mod}3),$ also
$2(n-2)(n-3)=0(\operatorname{mod}3).$ Thus, for $\left\vert \left[ 4^{1} \right]_{i} \right\vert=1$ and $\left\vert \left[  2^{1}, 1^{1} \right] _{i}^{(4)}  \right\vert=1,$ we have $\left\vert \left[ 3^{1}, 2^{1} \right]_{i}^{(2)} \right\vert=1(\operatorname{mod}3),$ while for
$\left\vert \left[ 5^{1} \right] _{i} \right\vert=1$ we get $\left\vert \left[ 3^{1}, 2^{1} \right]_{i}^{(2)} \right\vert=0(\operatorname{mod}3).$ By Lemma
\ref{L3}, $\left\vert \left[ 3^{1}, 2^{1} \right]_{i}^{(2)} \right\vert\leq2$ for all $i\in I.$ Hence, for $\left\vert \left[ 4^{1} \right]_{i} \right\vert=1$ and
$\left\vert \left[ 4^{1}, 1^{1} \right] _{i}^{(4)} \right\vert=1,$ we have $\left\vert \left[ 3^{1}, 2^{1} \right]_{i}^{(2)} \right\vert=1,$ and for $\left\vert \left[ 5^{1} \right] _{i} \right\vert=1$ we
get $\left\vert \left[ 3^{1}, 2^{1} \right]_{i}^{(2)} \right\vert=0.$ Again, by Lemma \ref{L3}, we have the same
conclusion for $\left\vert \left[ 3^{1}, 2^{1} \right]_{i}^{(3)} \right\vert.$ Thus,
\begin{equation}
\left\vert \left[ 3^{1}, 2^{1} \right]_{i}^{(3)} \right\vert=\left\vert \left[ 3^{1}, 2^{1} \right]_{i}^{(2)} \right\vert\label{gama}%
\end{equation}
for all $i\in I.$\bigskip

\noindent Substituting into (\ref{g2}) for $\left\vert \left[ 2^{1}, 1^{1} \right] _{i}^{(1)} \right\vert,\left\vert \left[  1^{3} \right] _{i}  \right\vert,$ and
$\left\vert \left[ 2^{1}, 1^{2} \right] _{i}^{(1)} \right\vert$ from Theorem \ref{7} yields $\left\vert \left[ 3^{1}, 1^{2} \right] _{i}^{(1)} \right\vert+2\left\vert \left[  2^{1}, 1^{3} \right] _{i}^{(1)}  \right\vert+\left\vert \left[ 2^{2}, 1^{1} \right] _{i}^{(2)} \right\vert+2\left\vert \left[ 2^{2}, 1^{1} \right] _{i}^{(1)} \right\vert%
=4(n-1)(n-2)-8(n-2)=4(n-2)(n-3).$ If $\left\vert \left[ 4^{1} \right]_{i} \right\vert=1$ or $\left\vert \left[ 4^{1}, 1^{1} \right] _{i}^{(4)} \right\vert=1,$ then in
this case $\left\vert \left[ 3^{1}, 2^{1} \right]_{i}^{(2)} \right\vert=1$, and by (\ref{r1}), (\ref{r2}), we get
\[
2\left\vert \left[ 2^{1}, 1^{3} \right] _{i}^{(1)} \right\vert+2\left\vert \left[ 2^{2}, 1^{1} \right] _{i}^{(1)} \right\vert=4(n-2)(n-3)-4n+11+\left\vert \left[  4^{1}, 1^{1} \right] _{i}^{(1)}  \right\vert.
\]
Thus, $\left\vert \left[  4^{1}, 1^{1} \right] _{i}^{(1)}  \right\vert$ is odd for $\left\vert \left[ 4^{1} \right]_{i} \right\vert=1$ and for $\left\vert \left[ 4^{1}, 1^{1} \right] _{i}^{(4)} \right\vert=1.$
However, $%
%TCIMACRO{\dsum \limits_{i\in I}}%
%BeginExpansion
{\displaystyle\sum\limits_{i\in I}}
%EndExpansion
\left\vert \left[ 4^{1}, 1^{1} \right] _{i}^{(1)} \right\vert=%
%TCIMACRO{\dsum \limits_{i\in I}}%
%BeginExpansion
{\displaystyle\sum\limits_{i\in I}}
%EndExpansion
\left\vert \left[ 4^{1}, 1^{1} \right] _{i}^{(4)} \right\vert$ and $\left\vert \left[ 4^{1}, 1^{1} \right] _{i}^{(4)} \right\vert\leq1.$ Therefore, for all $i\in I,$%
\begin{equation}
\left\vert \left[ 4^{1}, 1^{1} \right] _{i}^{(4)} \right\vert=\left\vert \left[ 4^{1}, 1^{1} \right] _{i}^{(1)} \right\vert,\label{beta}%
\end{equation}
which in turn implies $\left\vert \left[ 4^{1} \right]_{i} \right\vert=0$ for all $i,$ hence $\left\vert \left[ 4^{1} \right] \right\vert=0$. Combining $\left\vert \left[ 3^{1} \right] \right\vert=0$ with
(\ref{a}) yields $\left\vert \left[ 3^{1}, 1^{1} \right] \right\vert=2n,$ which in turn implies $\left\vert \left[ 2^{2} \right] \right\vert=0,$ see (\ref{ee}). The
proof is complete.
\end{proof}

\subsection{Phase B}

\noindent In this subsection we deal with the number of codewords of
individual types of weight equal to $5$. First we will summarize results for
all $n=2(\operatorname{mod}3),$ then we concentrate on the case $n=5$. For
$n=2(\operatorname{mod}3),$ all these values are expressed as a function of
the number of codewords of type $[5^{1}].$ We point out that for some
$n=2(\operatorname{mod}3)$, there are two tilings of $R^{n}$ by crosses with
different number of codewords of type $[5^{1}].$ Hence, unlike with codewords
of weight equal to $3$ or $4,$ the values of $\left\vert [5^{1}]\right\vert
,\left\vert [4^{1},1^{1}]\right\vert ,\left\vert [3^{1},2^{1}]\right\vert
,\left\vert [3^{1},1^{2}]\right\vert ,$ $\left\vert [2^{1},1^{3}]\right\vert
,\left\vert [2^{2},1^{1}]\right\vert ,$ and $\left\vert [1^{5}]\right\vert $
do not depend only on the value of $n$ but also on a given tiling of
$R^{n\text{ }}$by crosses.

\begin{theorem}
\label{9}Let $n=2(\operatorname{mod}3),$ and $W$ in $\mathcal{T}_{\mathcal{L}%
}.$ Then the number of codewords of a given type with respect to $W$ is:
$\left\vert [4^{1},1^{1}]\right\vert =2n-\left\vert [5^{1}]\right\vert
,\left\vert [3^{1},2^{1}]\right\vert =2n-\left\vert [5^{1}]\right\vert
,\left\vert [3^{1},1^{2}]\right\vert =2n(n-3)+\left\vert [5^{1}]\right\vert ,
$ $\left\vert [2^{1},1^{3}]\right\vert =n(2n-6)+\left\vert [5^{1}]\right\vert
,\left\vert [2^{1},1^{3}]\right\vert \\=\frac{1}{3}2n(n-3)(2n-7)-\left\vert
[5^{1}]\right\vert ,$ and $\left\vert [1^{5}]\right\vert =\frac{1}{5}%
(2^{4}\binom{n}{4}-n(n-3)(3n-8)+\left\vert [5^{1}]\right\vert ).$
\end{theorem}

\begin{proof}
We have proved in the previous theorem that $\left\vert \left[ 4^{1} \right] \right\vert=0.$ Therefore, by (\ref{d}), we
have $\left\vert \left[  5^{1} \right]  \right\vert+\left\vert \left[ 4^{1}, 1^{1} \right] \right\vert=2n.$ In addition, by Lemma \ref{L5}, it is $\left\vert \left[ 3^{1}, 2^{1} \right]   \right\vert=\left\vert \left[ 4^{1} \right] \right\vert+\left\vert \left[ 4^{1}, 1^{1} \right] \right\vert,$ that is $\left\vert \left[ 3^{1}, 2^{1} \right]   \right\vert=\left\vert \left[ 4^{1}, 1^{1} \right] \right\vert=2n-\left\vert \left[  5^{1} \right]  \right\vert.$ It is $\left\vert \left[ 3^{1}, 1^{2} \right] \right\vert=2n(n-2)-\left\vert \left[  4^{1}, 1^{1} \right]  \right\vert,$
see(\ref{eee1}), hence $\left\vert \left[  5^{1} \right]  \right\vert+\left\vert \left[ 4^{1}, 1^{1} \right] \right\vert=2n$ implies $\left\vert \left[ 3^{1}, 1^{2} \right] \right\vert=2n(n-3)+\left\vert \left[  5^{1} \right]  \right\vert,$
while (\ref{eee2}) implies $\left\vert \left[ 2^{2}, 1^{1} \right] \right\vert=n(2n-6)+\left\vert \left[  5^{1} \right]  \right\vert.$ The value of
$\left\vert \left[  2^{1}, 1^{3} \right]  \right\vert$ is given by (\ref{i}). Finally, the value of $\left\vert \left[  1^{5} \right]  \right\vert$ follows from
(\ref{h}) after substituting for $\left\vert \left[  1^{3} \right]  \right\vert$ and $\left\vert \left[  1^{4} \right]  \right\vert$ from Theorem \ref{5}. The proof
is complete.\bigskip
\end{proof}

\noindent The next theorem determines the local values of \thinspace
$\left\vert \left[  3^{1}, 2^{1} \right]  \right\vert,...,\left\vert \left[  1^{5} \right]  \right\vert$ with respect to $\left\vert \left[  5^{1} \right] _{1} \right\vert$.

\begin{theorem}
\label{10}Let $n=2(\operatorname{mod}3).$ Then for each codeword in
$\mathcal{T}_\mathcal{L}$ we have: If $\left\vert \left[ 5^{1} \right] _{i} \right\vert=1,$ then $\left\vert \left[  4^{1}, 1^{1} \right] _{i} \right\vert=0$,$\left\vert \left[  3^{1}, 2^{1} \right] _{i} \right\vert=0,$
$\left\vert \left[ 3^{1}, 1^{2} \right] _{i}^{(1)} \right\vert=2\left\vert \left[ 3^{1}, 1^{2} \right] _{i}^{(3)} \right\vert=2(n-2),\left\vert \left[ 2^{2}, 1^{1} \right] _{i}^{(2)} \right\vert%
=2\left\vert \left[ 2^{2}, 1^{1} \right] _{i}^{(1)} \right\vert=2(n-2),\left\vert \left[ 2^{1}, 1^{3} \right] _{i}^{(1)} \right\vert=3\left\vert \left[  2^{1}, 1^{3} \right] _{i}^{(2)}  \right\vert%
=(n-2)(2n-9).$ If $\left\vert \left[ 5^{1} \right] _{i} \right\vert=0$ , then $\left\vert \left[ 4^{1}, 1^{1} \right] _{i}^{(4)} \right\vert=\left\vert \left[ 4^{1}, 1^{1} \right] _{i}^{(1)} \right\vert=1,\left\vert \left[ 3^{1}, 2^{1} \right]_{i}^{(3)} \right\vert \\=\left\vert \left[ 3^{1}, 2^{1} \right]_{i}^{(2)} \right\vert=1,\left\vert \left[ 3^{1}, 1^{2} \right] _{i}^{(1)} \right\vert=2\left\vert \left[ 3^{1}, 1^{2} \right] _{i}^{(3)} \right\vert=2(n-3),$ $\left\vert \left[ 2^{2}, 1^{1} \right] _{i}^{(2)} \right\vert=2\left\vert \left[ 2^{2}, 1^{1} \right] _{i}^{(1)} \right\vert%
=2(n-3),\left\vert \left[ 2^{1}, 1^{3} \right] _{i}^{(1)} \right\vert=3\left\vert \left[  2^{1}, 1^{3} \right] _{i}^{(2)}  \right\vert=(n-3)(2n-7\dot{)}.$ In both cases
$\left\vert \left[ 1^{5} \right] _{i}  \right\vert=$ $\frac{1}{4}(8\binom{n-1}{3}-\left\vert \left[ 2^{1}, 1^{3} \right] _{i} \right\vert-\frac{10}{3}(n-2)(n-3)).$
\end{theorem}

\begin{proof}
Let $\left\vert \left[ 5^{1} \right] _{i} \right\vert=1.$ Then by, Lemma \ref{L5}, $\left\vert \left[ 3^{1}, 2^{1} \right]_{i}^{(3)} \right\vert=0,$ and by
(\ref{gama}) $\left\vert \left[ 3^{1}, 2^{1} \right]_{i}^{(2)} \right\vert=0;$ also $\left\vert \left[ 4^{1}, 1^{1} \right] _{i}^{(4)} \right\vert=0,$ and (\ref{gama})
implies $\left\vert \left[ 4^{1}, 1^{1} \right] _{i}^{(i)} \right\vert=0.$ Hence $\left\vert \left[ 3^{1}, 2^{1} \right]   \right\vert=\left\vert \left[ 4^{1}, 1^{1} \right] \right\vert=0.$ Using the same arguments
in the case $\left\vert \left[ 5^{1} \right] _{i} \right\vert=0$ yields $\left\vert \left[ 4^{1}, 1^{1} \right] _{i}^{(4)} \right\vert=\left\vert \left[ 4^{1}, 1^{1} \right] _{i}^{(1)} \right\vert%
=1,\left\vert \left[ 3^{1}, 2^{1} \right]_{i}^{(3)} \right\vert=\left\vert \left[ 3^{1}, 2^{1} \right]_{i}^{(2)} \right\vert=1.$ With this in hand, the values of
$\left\vert \left[  3^{1}, 1^{2} \right] _{i} \right\vert$ and $\left\vert \left[  2^{2}, 1^{1} \right] _{i} \right\vert$ follow from (\ref{r3}) and (\ref{r2}),
while the values of $\left\vert \left[  2^{1}, 1^{3} \right] _{i} \right\vert$ are obtained from (\ref{lambda}) and
(\ref{lambda1}). Finally, to determine $\left\vert \left[  1^{5} \right] _{i} \right\vert$ it suffices to substitute
into (\ref{h1}). The proof is complete.\bigskip
\end{proof}

\noindent We showed above, that the number of codewords in $\mathcal{T}_\mathcal{L}$ of
weight $3$ and $4$ does depends only on $n,$ while the number of
codewords of weight $5$ depends also on the tiling $\mathcal{L}$.
However, for $n=5$ also all these values are constant.

\begin{theorem}
\label{11} If $n=5,$ then $\left\vert \left[ 3^{1} \right] \right\vert=\left\vert \left[ 4^{1} \right] \right\vert=\left\vert \left[ 2^{2} \right] \right\vert=\left\vert \left[  5^{1} \right]  \right\vert=0,\left\vert \left[  2^{1}, 1^{1} \right]  \right\vert\\=\left\vert \left[  1^{3} \right]  \right\vert=\left\vert \left[ 3^{1}, 1^{1} \right] \right\vert=\left\vert \left[  1^{4} \right]  \right\vert=\left\vert \left[ 4^{1}, 1^{1} \right] \right\vert=\left\vert \left[ 3^{1}, 2^{1} \right]   \right\vert=10,$
$\left\vert \left[ 3^{1}, 1^{2} \right] \right\vert=\left\vert \left[ 2^{2}, 1^{1} \right] \right\vert=\left\vert \left[  2^{1}, 1^{3} \right]  \right\vert=20,$ while $\left\vert \left[  2^{1}, 1^{2} \right]  \right\vert=30,$ and $\left\vert \left[ 1^{5} \right]  \right\vert=2,$ and $\left\vert \left[ 1^{5} \right] _{i}  \right\vert=1$
for all $i\in I.$
\end{theorem}

\begin{proof}
 \hfill The \hfill values \hfill of $\left\vert \left[ 3^{1}\right] \right\vert,\left\vert \left[  2^{1}, 1^{1} \right]  \right\vert,\left\vert \left[  1^{3} \right]  \right\vert,\left\vert \left[ 4^{1} \right] \right\vert,\left\vert \left[ 3^{1}, 1^{1} \right] \right\vert,\left\vert \left[ 2^{2} \right] \right\vert,   \\ \left\vert \left[  2^{1}, 1^{2} \right]  \right\vert,$ and $\left\vert \left[  1^{4} \right]  \right\vert$ are obtained from Theorem \ref{6},
\ref{7}, and \ref{8} by substituting $n=5$. The other values depend on the
value of $\left\vert \left[  5^{1} \right]  \right\vert,$ see Theorem \ref{9}. We will prove that there is no
codeword of type $[5^{1}],$ i.e., that $\left\vert \left[  5^{1} \right]  \right\vert=0.$ In order to do it we have to
consider not only local equalities but also so-called double-local equalities
for the individual type of codewords. To be able to introduce these we need
one more piece of notation. Let $\mathcal{K}$ be a set of codewords. Then, for $i,j\in I,$ we denote by
$\left\vert  \mathcal{K}_{ij} \right\vert$ the number of codewords $K$ in $\mathcal{K}$ so that $K_{i}\neq0\neq
K_{j}.$ For an ordered pair $(i,j)$ we denote by $\left\vert \mathcal{K}_{ij}^{(a,b)} \right\vert$ the
number of codewords $K$ in $\mathcal{K}$ with $K_{i}=a$ and $K_{j}=b.$\bigskip

\noindent Substituting $n=5$ into Theorem \ref{9} and into Theorem \ref{10}
yields $\left\vert \left[ 1^{5} \right]  \right\vert=2+\frac{\left\vert \left[  5^{1} \right]  \right\vert}{5},$ and $\left\vert \left[  1^{5} \right] _{i} \right\vert>0$ for all $i\in I,$
respectively. Assume by contradiction that $\left\vert \left[  5^{1} \right]  \right\vert>0.$ Then we have $\left\vert \left[ 1^{5} \right]  \right\vert
\geq3,\,$\ and$\ $at least two codewords of type $[1^{5}]$ have to coincide
in at least two coordinates. Thus, there have to \ be signed coordinates
$i,j\in I$ such that $\left\vert \left[ 1^{5} \right]_{ij}  \right\vert\geq2.$ To reject the assumption of $\left\vert \left[  5^{1} \right]  \right\vert>0$
we prove that $\left\vert \left[ 1^{5} \right] _{ij} \right\vert\leq1$ for all $i,j\in I.$ We will start by setting
several double-local equalities.\bigskip

\noindent For each $i,j\in I$ there is a unique word $V$ in $Z^{5}$ of type
$[1^{2}]$ with $V_{i}=V_{j}=1.$ Therefore, see also the explanation to
(\ref{pp}),%

\begin{equation}
\left\vert \left[ 2^{1}, 1^{1} \right] _{ij} \right\vert+\left\vert \left[ 1^{3} \right] _{ij} \right\vert=1\label{x0}%
\end{equation}

\noindent For each $i,j\in I$ there are six words $V$ in $Z^{5}$ of type
$[1^{3}]$ with $V_{i}=V_{j}=1.$ Therefore, see also (\ref{c}),%

\begin{equation}
\left\vert \left[ 1^{3} \right] _{ij} \right\vert+\left\vert \left[ 2^{1}, 1^{2} \right] _{ij} \right\vert+2\left\vert \left[ 1^{4} \right] _{ij} \right\vert=6\label{x1}%
\end{equation}

\noindent Clearly, it is not difficult to see that $\left\vert \left[ 2^{1}, 1^{2} \right] _{ij} \right\vert=\left\vert \left[ 2^{1}, 1^{2} \right] _{ij}^{(2,1)} \right\vert+\left\vert \left[ 2^{1}, 1^{2} \right] _{ij}^{(1,2)} \right\vert \\+ \left\vert \left[ 2^{1}, 1^{2} \right] _{ij}^{(1,1)} \right\vert\leq3.$ Indeed, $\left\vert \left[ 2^{1}, 1^{2} \right] _{ij}^{(2,1)} \right\vert>1$
($\left\vert \left[ 2^{1}, 1^{2} \right] _{ij}^{(1,2)} \right\vert>1$) would imply that there are in $\mathcal{T}_\mathcal{L}$ two codewords
of type $[2^{1}, 1^{2}]$ of distance $2,$ a contradiction. Finally,
$\left\vert \left[ 2^{1}, 1^{2} \right] _{ij}^{(1,1)} \right\vert>1$ would imply, that there are in $\mathcal{T}_\mathcal{L}$ two codewords
of type $[2^{1}, 1^{2}]$ coinciding in two coordinates where their value is
$1,$ say $A=(1,1,2,0,0),$ and $B=(1,1,-2,0,0)$ or $B=(1,1,0,2,0).$ However, in
the former case $B-A=(0,0,4,0,0),$ and in the latter $B-A=(0,0,-2,2,0)$, thus
$B$ would be with respect to $A$ a codeword of type $[4^{1}],$ and of type
$[2^{2}],$ respectively, which contradicts $\left\vert \left[ 4^{1} \right] \right\vert=\left\vert \left[ 2^{2} \right] \right\vert=0,$ see Theorem \ref{8}.
Combining (\ref{x0}) with (\ref{x1}) we get:\bigskip

\noindent\textbf{Claim B. }For all $i,j\in I,$ if $\left\vert \left[ 1^{3} \right] _{ij} \right\vert=1,$ then $\left\vert \left[ 1^{4} \right] _{ij} \right\vert \geq1,$ and for $\left\vert \left[ 1^{3} \right] _{ij} \right\vert=0$, we get $\left\vert \left[ 1^{4} \right] _{ij} \right\vert\geq2.$\bigskip

\noindent Now we state a double-local equality for codewords of type
$[1^{4}].$ For each $i,j\in I,$ there are twelve words $V$ in $Z^{5}$ of
type $[1^{4}]$ with $V_{i}=V_{j}=1,$ using arguments similar to proving
(\ref{h}) yields:%
\begin{equation}
\begin{gathered}
\left\vert \left[ 1^{3} \right] _{ij}^{(0,1)} \right\vert+\left\vert \left[ 1^{3} \right] _{ij}^{(1,0)} \right\vert+4\left\vert \left[ 1^{3} \right] _{ij} \right\vert+\left\vert \left[ 1^{4} \right] _{ij} \right\vert+\left\vert \left[ 2^{1}, 1^{3} \right] _{ij} \right\vert+\\ 3\left\vert \left[ 1^{5} \right] _{ij} \right\vert%
=12,\label{x2}%
\end{gathered}
\end{equation}

\noindent where $\left\vert \left[ 1^{3} \right] _{ij}^{(0,1)} \right\vert$ is the number of codewords $C$ of type
$[1^{3}]$ such that $C_{i}=C_{-i}=0,$ and $C_{j}=1.$ For each $i,j\in I,$
we have $\left\vert \left[ 1^{3} \right] _{ij} \right\vert\leq1,$ otherwise there would be two codewords of type
$[1^{3}]$ at distance less than $3.$ Therefore, for all $i,j\in I,$ there
is at most one codeword $V$ of type $[1^{3}]$ with $V_{i}=V_{j}=1,$ or
$V_{-i}=V_{j}=1,$ or $V_{i}=V_{-j}=1.$ As $\left\vert \left[  1^{3} \right] _{i}  \right\vert=3,$ we get that both
$\left\vert \left[ 1^{3} \right] _{ij}^{(0,1)} \right\vert$ $\geq2$ and $\left\vert \left[ 1^{3} \right] _{ij}^{(1,0)} \right\vert\geq2$ for $\left\vert \left[ 1^{3} \right] _{ij} \right\vert=0,$ and both
$\left\vert \left[ 1^{3} \right] _{ij}^{(0,1)} \right\vert$ $\geq1$ and $\left\vert \left[ 1^{3} \right] _{ij}^{(1,0)} \right\vert\geq1$ if $\left\vert \left[ 1^{3} \right] _{ij} \right\vert=1.$\bigskip

\noindent In aggregate, if $\left\vert \left[ 1^{3} \right] _{ij} \right\vert=1,$ we get $\left\vert \left[ 1^{3} \right] _{ij}^{(0,1)} \right\vert+\left\vert \left[ 1^{3} \right] _{ij}^{(1,0)} \right\vert+4\left\vert \left[ 1^{3} \right] _{ij} \right\vert\geq6,$ and $\left\vert \left[ 1^{3} \right] _{ij}^{(0,1)} \right\vert+\left\vert \left[ 1^{3} \right] _{ij}^{(1,0)} \right\vert+4\left\vert \left[ 1^{3} \right] _{ij} \right\vert\geq4$ for
$\left\vert \left[ 1^{3} \right] _{ij} \right\vert=0.$ 

\noindent Substituting to (\ref{x2}) for $\left\vert \left[ 1^{4} \right] _{ij} \right\vert$ from Claim B implies

\begin{equation*}
\begin{gathered}
\left\vert \left[ 2^{1}, 1^{3} \right] _{ij} \right\vert+3\left\vert \left[ 1^{5} \right] _{ij} \right\vert\leq5\text{ for }  \left\vert \left[ 1^{3} \right] _{ij} \right\vert=1, \text{ and } \\ \left\vert \left[ 2^{1}, 1^{3} \right] _{ij} \right\vert%
+3\left\vert \left[ 1^{5} \right] _{ij} \right\vert\leq6\text{ for }\left\vert \left[ 1^{3} \right] _{ij} \right\vert=0.
\end{gathered}
\end{equation*} \bigskip

\noindent Thus, $\left\vert \left[ 1^{5} \right] _{ij} \right\vert\leq1$ for $\left\vert \left[ 1^{3} \right] _{ij} \right\vert=1.$ To prove that $\left\vert \left[ 1^{5} \right] _{ij} \right\vert\leq1$
also when $\left\vert \left[ 1^{3} \right] _{ij} \right\vert=0,$ we will show that in this case $\left\vert \left[ 2^{1}, 1^{3} \right] _{ij} \right\vert>0.$ If
$\left\vert \left[ 1^{3} \right] _{ij} \right\vert=0,$ then $\left\vert \left[ 2^{1}, 1^{1} \right] _{ij} \right\vert=1,$ see (\ref{p1}). Let $B$ be the codeword of type
$\left[  2^{1},1^{1}\right]  $ having non-zero the $i$-th and the $j$-th sign
coordinate$.$ Then either $B_{i}=2$ and $B_{j}=1,$ or $B_{i}=1,$ and
$B_{j}=2.$ Assume that the former is the case. In $Z^{5}$ there are six words
$V$ of type $[2^{1}, 1^{2}]$ with $V_{i}=2,$ and $V_{j}=1.$ Each of these is
covered either (i) by a codeword $W$ of type $[2^{1}, 1^{1}]$ with $W_{i}%
=2,W_{j}=0;$ or (ii) by a codeword $W$ of type $[2^{1}, 1^{1}]$ with
$W_{i}=2,W_{j}=1;$ or (iii) by a codeword $W$ of type $[1^{3}]$ with
$W_{i}=1,W_{j}=1;$ or (iv) by a codeword $W$ of type $\left[  2^{1}, 1^{2}\right]  $ with $W_{i}=2,W_{j}=1;$ or (v) by a codeword $W$ of type
$\left[  3^{1}, 1^{2}\right]  $ with $W_{i}=3,W_{j}=1;$ or (vi) by a codeword
$W$ of type $\left[  2^{2}, 1^{1}\right]  $ with $W_{i}=2,W_{j}=2;$ or
(vii) by a codeword $W$ of type $\left[  2^{2}, 1^{1}\right]  $ with
$W_{i}=2,W_{j}=1;$ or (viii) by a codeword $W$ of type $\left[  2^{1},
1^{3}\right]  $ with $W_{i}=2,W_{j}=1.$ Using our notation we can write%

\begin{equation*}
\begin{gathered}
\left\vert \left[ 2^{1}, 1^{1} \right] _{ij}^{(2,0)} \right\vert+6\left\vert \left[ 2^{1}, 1^{1} \right] _{ij}^{(2,1)} \right\vert+\left\vert \left[ 1^{3} \right] _{ij} \right\vert+\left\vert \left[ 2^{1}, 1^{2} \right] _{ij}^{(2,1)} \right\vert+\left\vert \left[ 3^{1}, 1^{2} \right] _{ij}^{(3,1)} \right\vert\\ +\left\vert \left[ 2^{2}, 1^{1} \right] _{ij}^{(2,2)} \right\vert+\left\vert \left[ 2^{2}, 1^{1} \right] _{ij}^{(2,1)} \right\vert+\left\vert \left[ 2^{1}, 1^{3} \right] _{ij}^{(2,1)} \right\vert=6.
\end{gathered}
\end{equation*}

\noindent We have chosen $i,j$ so that $\left\vert \left[ 2^{1}, 1^{1} \right] _{ij}^{(2,1)} \right\vert=\left\vert \left[ 1^{3} \right] _{ij} \right\vert=0.$ Further, it
is easy to see that in each of the cases (i) and (iv)-(vi) there is at most
one codeword of each type, as otherwise we would have two codewords at
distance less than $3$. As to (vii), we have $\left\vert \left[ 2^{2}, 1^{1} \right] _{ij}^{(2,1)} \right\vert\leq1 $
as well, because if there were two codewords $W$ of type $[2^{2},1^{1}]$
with $W_{i}=2,$ and $W_{j}=1,$ then their difference would be a codeword
either of type $[4^{1}]$ or $[2^{2}],$ a contradiction. So we get
$\left\vert \left[ 2^{1}, 1^{3} \right] _{ij}^{(2,1)} \right\vert\geq1$ in this case, thus $\left\vert \left[ 2^{1}, 1^{3} \right] _{ij} \right\vert\geq1,$ and in
turn $\left\vert \left[ 1^{5} \right] _{ij} \right\vert\leq1$ also in the case $\left\vert \left[ 1^{3} \right] _{ij} \right\vert=0.$ The proof is complete.
\end{proof}

\subsection{Phase C}

\noindent In the previous subsection we proved that the $5$-neighborhood of
each codeword has the same quantitative properties. Now we prove that it has
also the same structure.\bigskip 

\noindent Let $V$ be a word in $Z^{5}.$ Then by $<V>$ \ we denote the
collection of words comprising $V,$ and the words obtained by cyclic shifts of
coordinates of $V$. Hence, e.g., $<(2,1,0,0,0)>=\{(2,1,0,0,0),(0,2,1,0,0),
$\newline$(0,0,2,1,0),$ $(0,0,0,2,1),(1,0,0,0,2)\}.\,\ $We note that $<V>$ contains five words except for the case when $V$ has all coordinates equal to the same number. Finally, we set $\pm<V>$ $\circeq$ $<V>\cup<-V>.$ By the canonical $5$-
neighborhood, or simply a canonical neighborhood, we mean the set of
words\newline$\{\pm<(2,1,0,0,0)>,\pm<(1,0,1,0,-1)>,\pm<(3,0,0,-1,0)>,\newline%
\pm<(2,0,1,0,1)>,\pm<(2,0,0,1,-1)>,\pm<(2,-1,-1,0,0)>,\newline\pm
<(1,1,1,-1,0)>,\pm<(4,0,-1,0,0)>,\pm<(3,0,2,0,0)>,\newline\pm<(3,0,0,1,1)>,\pm
<3,-1,0,0,-1)>,$ $\pm<(2,-2,0,0,1)>,\newline\pm<(2,0,-2,0,-1)>,\pm
<(2,-1,1,1,0)>,\pm<(2,0,-1,-1,1)>,$ \newline$\pm<(1,1,1,1,1)>\}.$ A simple
inspection shows that the number of words of individual types in the canonical
neighborhood coincides with the values given by Theorem \ref{11}. E.g.,
$\pm<(2,1,0,0,0)>$ is the set of ten words of type $[2^{1},1^{1}],$ ($\left\vert \left[ 2^{1}, 1^{1} \right] \right\vert=10$),
while $\pm<(1,0,1,0,-1)>$ comprises ten words of type $[1^{3}]$ ($\left\vert \left[ 1^{3} \right] \right\vert=10$).

\begin{theorem}
\label{12} Let $\mathcal{L}$ be a tiling of $R^{5}$ by crosses. Then, for
each\ codeword $W$ in $\mathcal{T}_\mathcal{L}$, $\ $the $5$-neighbourhood of $W$ is
congruent to the canonical one. Moreover, the $5$-neighborhood of $W$ is
uniquely determined by the set of codewords of type $[2^{1}, 1^{1}]$.
\end{theorem}

\begin{proof}
As in other proofs in this paper we assume w.l.o.g. that $W=O.$ Two words
$U=(u_{1},...,u_{5}),V=(v_{1},..,v_{5})$ will be called sign equivalent in the
$j$-th coordinate if $u_{j}v_{j}>0;\,\ $that is, they are sign equivalent if
$u_{j}\neq0\neq v_{j},$ and the two non-zero values have the same sign.

\underline{Codewords of type $[1^{5}].$} It was proved in Theorem \ref{11}
that $\left\vert \left[ 1^{5} \right]  \right\vert=2,$ and $\left\vert \left[ 1^{5} \right] _{i}  \right\vert=1$ for each $i\in I;$ i.e., the two codewords of
type $[1^{5}]$ differ in each coordinate (= are not sign equivalent in any
coordinate). That is, if $M$ is a codeword of type $[1^{5}],$ then $-M$ is
the other codeword of type $[1^{5}].$

\underline{Codewords of type $[2^{1}, 1^{1}].$} Let $\mathcal{B}$ be a set of
codewords of type $[2^{1}, 1^{1}],$ and $\mathcal{C}$ be the set of codewords of
type $[1^{3}]$. We know by Theorem \ref{6} that $\left\vert \mathcal{B}%
\right\vert =\left\vert \mathcal{C}\right\vert =10$. There are in total $10$
words of type $[1^{2}]$ that are sign equivalent in two coordinates with
$M$ and another $10$ words of type $[1^{2}]$ that are sign equivalent in
two coordinates with $-M.$ Each word of type $[1^{2}]$ is covered by a
codeword in $\mathcal{B\cup C}$, thus each of these $20$ words is covered by a
codeword in $\mathcal{B\cup C}$. If a codeword $C$ in $\mathcal{C}$ covered
two of these $20$ words, then $C$ would be sign equivalent in three
coordinates with $M$ or $-M,$ and the distance of $C$ to one of $M$ or $-M$
would be less than $3$. Therefore, each codeword in $\mathcal{C}$ covers at
most one of these $20$ words of type $[1^{2}]$. Hence, each codeword in
$\mathcal{B}$ has to cover one of these $20$ words, that is, each codeword of
type $[2^{1}, 1^{1}]$ is sign equivalent in both non-zero coordinates either with
the codeword $M$ or with the codeword $-M$. We know by Theorem \ref{6} that,
for each $i\in I,$ $\left\vert \left[ 2^1, 1^1 \right]_i^{(2)} \right\vert=\left\vert \left[ 2^{1}, 1^{1} \right] _{i}^{(1)} \right\vert.$ Thus, five codewords in
$\mathcal{B}$ are sign equivalent in two coordinates with $M,$ the other five
with $-M$.\bigskip

\noindent It turns out that graph theory has a very suitable language to
describe the structure of the set $\mathcal{B}$. Let $G$ be a graph with the
vertex set $I,$ the set of signed coordinates, and the edges of $G$ be all
pairs of vertices in $I$ except for $\{i,-i\},$ $i=1,...,5.$ Thus $G$ is a
complete graph on $10$ vertices, $K_{10},$ with a $1$-factor (=perfect
matching) removed. We denote this one factor by $M.$ So $G=K_{10}-M.$ In
$Z^{5}$ there are forty words of type $\left[ 1^{2}\right]  .$ In a
natural way each word of type $\left[  1^{2}\right]  $ is associated with
an edge of $G.$ If $V$ is a word of type $\left[  1^{2}\right]  $ with
$V_{i}=V_{j}=1,$ (and then $V_{k}=0$ for all $k\in I-\{i,j\}$) then we assign
to $V$ the edge $\{i,j\}$ of $G.$ So there is a one-to-one correspondence
between words of type $\left[  1^{2}\right]  $ and the edges of $G.$ In
addition, each codeword $W\in\mathcal{B}$ is associated with the edge (word of
type $[1^{2}])$ covered by $W.$ The set of words of type $\left[  
1^{2}\right]  $ covered by codewords in $\mathcal{B}$ will be denoted by
$\mathcal{B}^{\ast}.$ The condition $\left\vert \left[ 2^1, 1^1 \right]_i^{(2)} \right\vert=\left\vert \left[ 2^{1}, 1^{1} \right] _{i}^{(1)} \right\vert=1$ implies that
the words in $\mathcal{B}^{\ast}$ form a $2$-factor, say $F.$ It follows from
the above discussion that $F$ consists of two cycles of length $5$ such that
both cycles contain exactly one vertex of each edge in the matching $M$. Let
$C$ be the set of five words in $\mathcal{B}^{\ast}$ constituting one of the
two $5$-cycles in $F$. Clearly, by suitably permuting the order of coordinates
of each codeword in $\mathcal{T}_\mathcal{L}$ and/or changing a sign of a coordinate for
each codeword in $\mathcal{T}_\mathcal{L}$ maps $C$ onto $<(1,1,0,0,0)>,$ and the
codewords covering words in $C$ onto $<(2,1,0,0,0)>$. $\ $By Claim \ref{0} the
above transformation is a congruence mapping. Therefore we will assume that
$<(2,1,0,0,0)>$ are codewords in $\mathcal{T}_\mathcal{L}$.\bigskip

\noindent We will show that by choosing $<(2,1,0,0,0)>\in\mathcal{T}_\mathcal{L}$ all the
other codewords will be uniquely determined. First of all $<(2,1,0,0,0)>\in
\mathcal{T}_\mathcal{L}$ implies that the codewords of type $\left[  1^{5}\right]  $
are $M=(1,1,1,1,1)$ and $-M.$\bigskip

\noindent If one of the two $5$-cycles of $\mathcal{B}^{\ast}$ is
$<(1,1,0,0,0)>,$ then the other one is formed by codewords having both
non-zero coordinates negative. There are four non-isomorphic ways how to
choose it. It is either the $5$-cycle comprising \ the edges corresponding to
$<(-1,-1,0,0,0)>,$ or to $<(-1,0,0,-1,0)>,$ or\newline%
$\{(-1,-1,0,0,0),(0,-1,0,0,-1),(0,0,-1,0,-1),(0,0,-1,-1,0),$~\newline%
$\ (-1,0,0,-1,0)\}$, or\newline%
$\{(-1,-1,0,0,0),(0,-1,-1,0,0),(0,0,0,-1,-1),(0,0,-1,0,-1),$~\newline%
$\ (-1,0,0,-1,0)\}.$\bigskip

\noindent The graph consisting of the edges of the $2$-factor $F$ and the
matching $M$ is a cubic graph. The four cubic graph corresponding to the cases
described above are depicted in Figure 1. \ The first one is the prism on $10$
vertices, the second is the Petersen graph, the labels for the other two are
taken from \cite{Bryant}.%

%TCIMACRO{\FRAME{ftbpFU}{5.3852in}{1.5264in}{0pt}{\Qcb{Corresponding cubic
%graphs}}{}{grafy_v4.eps}{\special{ language "Scientific Word";
%type "GRAPHIC";  maintain-aspect-ratio TRUE;  display "USEDEF";
%valid_file "F";  width 5.3852in;  height 1.5264in;  depth 0pt;
%original-width 6.3027in;  original-height 1.7625in;  cropleft "0";
%croptop "1";  cropright "1";  cropbottom "0";
%filename 'Grafy_v4.eps';file-properties "XNPEU";}}}%
%BeginExpansion

\begin{figure}[ptb]\begin{center}
\includegraphics[scale=0.85]{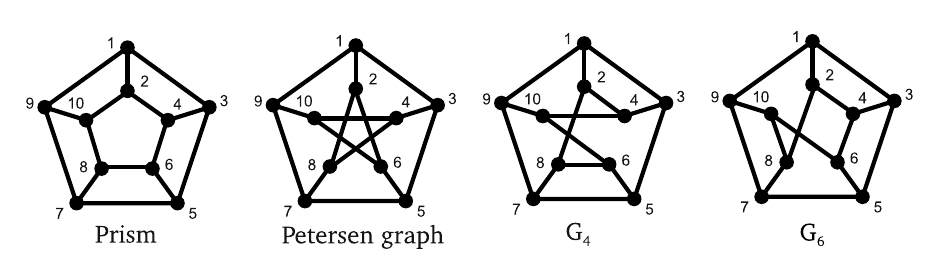}
\caption{Corresponding cubic graphs}
\end{center}\end{figure} 

%\begin{figure}[ptb]\begin{center}
%\includegraphics[
%natheight=1.7625in, natwidth=6.3027in, height=1.5264in, width=5.3852in]
%{Grafy_v4.pdf}%
%\caption{Corresponding cubic graphs}
%\end{center}\end{figure}%

%EndExpansion

\underline{Codewords of type $[1^{3}]$}. Each codeword of type $[1^{3}]$
covers three words of type $[1^{2}],$ so we associate with each codeword of
type $[1^{3}]$ a triangle (cycle of length $3$) in the graph $G$. As
mentioned many times, all words of type $[1^{2}]$ are covered by codewords
in $\mathcal{B\cup C}.$ Therefore, the triangles corresponding to codewords in
$\mathcal{C}$ have to form an edge decomposition of the complement of the
cubic graph consisting of the edges of the $2$-factor $F$ and the matching
$M.$\bigskip

\noindent It is known, see \cite{Bryant}, that the complement of the two cubic
graphs $G_{4}$ and $G_{6}$ is not decomposable into triangles. Therefore, in
the case of $G_{4}$ and $G_{6}$ it is impossible to choose the codewords of
type $[1^{3}].$ Thus, we are left with the prism and the Petersen graph.

(i) The prism. Then $\mathcal{B}^{\ast}=\pm<(1,1,0,0,0)>.$ There are two
different ways how to choose codewords of type $[1^{3}]  $
(how to decompose the complement of the prism into triangles). Either

(ia) $\mathcal{C}=\pm<(1,0,1,0,-1)>,$ or

(ib) $C=\pm<(1,0,1,-1,0)>.$

\noindent These two decomposition are isomorphic but we will we need to
consider both of them, as the automorphism of the graph $G,$ which maps one
decomposition on the other, maps the set $\pm<(2,1,0,0,0)>$ of codewords of
type $[2^{1}, 1^{1}]$ on the set $\pm<(1,2,0,0,0)>.$\bigskip

(ii) The Petersen graph.

\noindent Then $\mathcal{B}^{\ast}=<(1,1,0,0,0)>\cup<(-1,0,-1,0,0)>.$ There
are six different ways how to decompose the complement of the Petersen graph
into triangles. One of them corresponds to the following codewords of type
$[1^{3}]  :$

(iia) $\mathcal{C=<(}1,-1,1,0,0)>\cup<-1,-1,0,1,0)>,$ while the other five are
isomorphic to:

(iib) $\mathcal{C}=\{(1,-1,1,0,0),(1,0,-1,1,0),(-1,1,0,0,1),\newline%
(0,1,0,1,0,-1),(0,0,1,-1,1),(1,0,0,-1,-1),(-1,-1,0,1,0),\newline%
(-1,0,1,0,-1),(0,1,-1,-1,0),(0,-1,-1,0,1)\}.$ Here we do not need to consider
all $5$ decompositions, as it is possible to prove that this case is not a
viable one whether $<(2,1,0,0,0)>\in\mathcal{T}_\mathcal{L}$ or $<(1,2,0,0,0)>\in
\mathcal{T}_\mathcal{L}.$

\underline{Codewords of type $\left[  1^{4}\right]  .$} Let $\mathcal{H}$
be the set of ten codewords in $\mathcal{T}_\mathcal{L}$ of type $\left[  1^{4}\right]
.$ Clearly, $\mathcal{H}\subset\mathcal{H}^{\ast}$ where $\mathcal{H}^{\ast}$
is constructed as follows: First let $\mathcal{H}^{\ast}$ be the set of all
eighty words in $Z^{5}$ of type $\left[  1^{4}\right]  .$ As any codeword
in $\mathcal{H}$ has to have a distance at least $3$ from both codewords of
type $\left[  1^{5}\right]  ,$ and all ten codewords of type $\left[
1^{3}\right]  ,$ we delete from $\mathcal{H}^{\ast}$ all codewords that
coincide with a codeword in all four non-zero coordinates, or with a codeword
in $\mathcal{C}$ in three non-zero coordinates. After these two procedures
there are exactly thirty words left in $\mathcal{H}^{\ast}.$\bigskip

\noindent There are twenty words of type $[1^{2}]$ with both non-zero
coordinates of the same sign. Ten of them are covered by codewords in
$\mathcal{B}$, the other ten by the codewords in $\mathcal{C}$. Further, by
Claim B stated in the proof of Theorem \ref{11}, it is $\left\vert \left[ 1^{4} \right] _{ij} \right\vert\geq2$ for $ij $
such that $\left\vert \left[ 2^{1}, 1^{1} \right] _{ij} \right\vert=1,$ and $\left\vert \left[ 1^{4} \right] _{ij} \right\vert\geq1$ for $\left\vert \left[ 1^{3} \right] _{ij} \right\vert=1$. Hence, to each
codeword $W$ in $\mathcal{B},$ there are in $\mathcal{H}$ at least two
codewords sign equivalent with $W$ in two coordinates. Since no codeword in
$\mathcal{H}$ has all non-zero coordinates of the same sign (otherwise its
distance to a codeword of type $[1^{5}]$ would be less than $3$), and the
codewords in $\mathcal{B}$ form two $5$-cycles, no codeword in $\mathcal{H}$
can be sign equivalent in two coordinates with three codewords in
$\mathcal{B}$. This in turn implies, because $\left\vert \mathcal{H}%
\right\vert =10,$ that\bigskip

\noindent\textbf{Claim C. }Each codeword in $\mathcal{H}$ has to be sign
equivalent in two coordinates with two codewords of type $[2^{1}, 1^{1}]$ and
with one word of type $[1^{2}]$ with both non-zero coordinates of the same
sign. In particular, each codeword in $\mathcal{H}$ has three coordinates of
the same sign.\bigskip

\noindent As in all cases $<(1,1,0,0,0)>\subset\mathcal{B}^{\ast}%
$,~$\mathcal{H}$ has to contain a codeword $W=(1,1,1,a,b)$ where exactly one
of $a,b$ equals $0$ and the other equals $-1,$ and all cyclic shifts of
coordinates of $W.$\bigskip

\noindent(ia) First we describe the set $\mathcal{H}^{\ast}.$ At the beginning
of the process$\newline$ $\mathcal{H}^{\ast}=\{\pm<(1,1,1,1,0)>,\pm
<(1,1,1,-1,0)>,\pm<(1,1,-1,1,0)>,\newline\pm<(1,-1,1,1,0)>,\pm
<(-1,1,1,1,0)>,\pm<(1,1,-1,-1,0)>,\newline\pm<(1,-1,1,-1,0)>,\pm
<(-1,1,1,-1,0)>\}.$ We have to remove from $\mathcal{H}^{\ast}$ words
$\pm<(1,1,1,1,0)>$ that have the distance from the codewords $\pm(1,1,1,1,1)$
of type $[1^{5}]$ less than $3.$ Next, in this case, the set $\mathcal{C}$
of codewords of type $[1^{3}]$ is $\pm<(1,0,1,0,-1)>.$ Therefore we need to
remove from $\mathcal{H}^{\ast}$ forty words at distance less than $3$ from
any codeword in $\mathcal{C}$. These words are $\pm<(1,\pm1,1,0,-1)>$ and
$\pm<(1,0,1,\pm1,-1)>.$ Thus, at the end of the process $\mathcal{H}^{\ast
}=\{\pm<(1,1,1,-1,0)>,\pm<(1,-1,1,1,0)>,\pm<(1,1,-1,-1,0,)>\}.$ Clearly, the
only way how to choose a set of ten codewords satisfying Claim C is to set
$\mathcal{H}=\pm<(1,1,1,-1,0)>.$\bigskip

(ib) We have $\mathcal{H}^{\ast}=\{\pm<(-1,1,1,1,0)>,\pm<(1,1,-1,1,0)>,$%
\linebreak$\pm<(1,1,-1,-1,0)>\}.$ Then a unique way how to fulfill Claim C is
to set $\mathcal{H}=\pm<(-1,1,1,1,0)>.$\bigskip

(iia) Let $H_{1}=<(1,1,1,-1,0)>,$ $H_{2}=<(-1,1,1,1,0)>,\newline
H_{3}=<(-1,1,-1,-1,0)>,H_{4}=<-1,-1,1,-1,0)>.$ Then $\mathcal{H}^{\ast}$ can
be expressed as$\newline$ $\mathcal{H}^{\ast}=%
%TCIMACRO{\dbigcup \limits_{i=1}^{4}}%
%BeginExpansion
{\displaystyle\bigcup\limits_{i=1}^{4}}
%EndExpansion
H_{i}\cup<(-1,1,1,-1,0)>\cup<(-1,-1,1,1,0)>.$\bigskip

\noindent There are four options how to choose $\mathcal{H}$ in this case.
Either $\mathcal{H}$=$H_{1}\cup H_{3},$ or $H_{1}\cup H_{4},$ or $H_{2}\cup
H_{3},$ or $H_{2}\cup H_{4}.$\bigskip

(iib) In this case all words in $\mathcal{H}^{\ast}$ that belong to
$<(1,1,1,a,b)>$ are:\newline$(1,1,1,-1,0),(1,1,1,0,-1),$
$(-1,1,1,1,0),(-1,0,1,1,1),(0,-1,1,1,1),$\newline$(1,-1,0,1,1),(1,1,0,-1,1).$
It is impossible to choose from this set three words of type
$(a,1,1,1,b),(a,b,1,1,1),$ and $(1,a,b,1,1)$ that would be pair-wise at
distance at least $3.$ Therefore in this case it is impossible to choose a
required set of codewords of type $[1^{4}],$ and we do not need to consider
this case any longer.

\underline{Codewords of type $[2^{1}, 1^{2}]$ and $[2^{1},1^{3}].$} In
$Z^{5}$ there are eighty words of type $[1^{3}].$ By (\ref{c}), ten of them
are covered by codewords in $\mathcal{C}$ of type $[1^{3}],$ forty of them
by codewords in $\mathcal{H}$ of type $[1^{4}],$ and the set $\mathcal{G}%
^{\ast\text{ }}$of the remaining thirty words covered by the set $\mathcal{G}$
of codewords of type $[2^{1}, 1^{2}]$. Clearly, if a codeword $W$ of type
$[2^{1}, 1^{2}]$ covers a word $V$ in $\mathcal{G}^{\ast},$ then we can see
the codeword $W$ as obtained from $V$ by multiplying one of the non-zero
coordinates of $V$ by two.\bigskip

\noindent Further, in $Z^{5}$ there are eighty words of type $[1^{4}].$ By
(\ref{h}), forty of them are covered by codewords in $\mathcal{C}$, ten of
them by codewords in $\mathcal{H},$ ten by codewords of type $[1^{5}], $
and the remaining twenty, belonging to the set $\mathcal{H}^{\ast}%
-\mathcal{H}$, by the set $\Lambda$ of codewords of type $[2^{1},1^{3}] $.
As above, if a codeword $W$ in $\Lambda$ covers a word $V$ in $\mathcal{H}%
^{\ast}-\mathcal{H}$, then we can see $W$ as obtained from $V$ by multiplying
one of the non-zero coordinates of $V$ by two.\bigskip

(ia) Since $\mathcal{H=\pm<(}1,1,1,-1,0)>,$ and $\mathcal{C=}\pm
<(1,0,1,0,-1)>,$ it is $\mathcal{G}^{\ast}=\{\pm<(1,0,1,1,0)>,\pm
<(1,-1,1,0,0)>,$\linebreak$\pm<(0,-1,1,1,0)>\},$ and $\mathcal{H}^{\ast
}-\mathcal{H}=\{\pm<(1,-1,1,1,0)>,$\linebreak$\pm<(1,1,-1,-1,0)>\}.$ We need
to consider two possible choices of codewords of type $[2^{1}, 1^{1}]$ covering
the other $5$-cycle $<(-1,-1,0,0,0)>$ of the $2$-factor $F.$ It is either
$<(-2,-1,0,0,0)>\in\mathcal{T}_\mathcal{L},$ or \linebreak$<(-1,-2,0,0,0)>\in\mathcal{T}_\mathcal{L}$.

In the former case consider the set of words/codewords.\medskip%

\begin{tabular}
[c]{rrrrrr}%
$B_{1}=($ & $-2,$ & $-1,$ & $0,$ & $0,$ & $0)$\\
$B_{2}=($ & $0,$ & $0,$ & $2,$ & $1,$ & $0)$\\
$V_{1}=($ & $1,$ & $-1,$ & $1,$ & $1,$ & $0)$\\
$V_{2}=($ & $-1,$ & $-1,$ & $1,$ & $1,$ & $0)$\\
$V_{3}=($ & $1,$ & $0,$ & $1,$ & $1,$ & $0)$\\
$V_{4}=($ & $0$ & $-1,$ & $1,$ & $1,$ & $0)$\\
$V_{5}=($ & $1,$ & $-1,$ & $1,$ & $0,$ & $0)$%
\end{tabular}
\medskip

\noindent Clearly, $B_{1},B_{2}$ are codewords of type $[2^{1}, 1^{1}],$ while
$V_{1},V_{2}\in\mathcal{H}^{\ast}-\mathcal{H}$, and $V_{3},V_{4},V_{5}%
\in\mathcal{G}^{\ast}$. To get a codeword $W_{j}$ covering the word $V_{j}%
,$\linebreak$j=1,...,5,$ we need to multiply a non-zero coordinate of $V_{j}$
by $2$. As there are five words $V_{j},$ and all of them have the fifth
coordinate equal to zero, two of them have to have the same non-zero
coordinate multiplied by $2$. It is possible only for $V_{2}$ and $V_{3}$ if
their forth coordinate is chosen, as otherwise the two resulting codewords
would be at distance less than $3$ (note that multiplying the first coordinate
of $V_{2}$ by $2$ results in a codeword at distance less than $3$ from
$B_{1}).$ Because of the distance to the codeword $B_{2},$ only the word
$V_{5}$ can have its third coordinate multiplied by $2$. This implies that
$V_{4}$ has to have its second coordinate multiplied by $2$ while for $V_{1}%
$the first one is the only choice.\bigskip

\noindent The same type of argument can be applied to a set of words/codewords
obtained by the above one by cyclically shifting coordinates of each
word/code\-word and/or multiplying all words/codewords by $-1.$ Therefore
setting\newline$\Lambda=\{\pm<(2,-1,1,1,0)>,\pm<(-1,-1,1,2,0)>\},$
and\newline$\mathcal{G}=\{\pm<(1,0,1,2,0)>,\pm<(0,-2,1,1,0)>,\pm
<(1,-1,2,0,0)>\},$ which is identical to\newline$\Lambda=\{\pm
<(2,-1,1,1,0)>,\pm<(2,0,-1,-1,1)>\},$and \newline$\mathcal{G}=\{\pm
<(2,0,1,0,1)>,\pm<(2,0,0,1,-1)>,\pm<(2,-1,-1,0,0)>\}$ is the unique choice so
that all codewords in $\mathcal{B}$, $\mathcal{G}$, and $\Lambda$ are
pair-wise at distance at least $3$.\bigskip

\noindent In the latter case $<(-1,-2,0,0,0)>\in\mathcal{B}$. We will
demonstrate that this is not a viable option. Let\medskip%

\begin{tabular}
[c]{rrrrrr}%
$B=($ & $0,$ & $0,$ & $-1,$ & $-2,$ & $0)$\\
$V_{1}=($ & $-1,$ & $1,$ & $-1,$ & $-1,$ & $0)$\\
$V_{2}=($ & $-1,$ & $0,$ & $-1,$ & $-1,$ & $0)$\\
$V_{3}=($ & $-1,$ & $1,$ & $-1,$ & $0,$ & $0)$\\
$V_{4}=($ & $0,$ & $1,$ & $-1,$ & $-1,$ & $0)$%
\end{tabular}
\medskip

\noindent It is easy to check that $V_{1}\in\pm<(1,-1,1,1,0)>\in
\mathcal{H}^{\ast}-\mathcal{H},$ and $V_{2}\in\pm<(1,0,1,1,0)>\in
\mathcal{G}^{\ast},V_{3}\in\pm<(1,-1,1,0,0)>\in\mathcal{G}^{\ast},$
\linebreak$V_{4}\in\pm<(0,-1,1,1,0)>\in\mathcal{G}^{\ast},$ while $B$ is a
codeword. Let $W_{j}$ be a codeword covering the word $V_{j},j=1,...,4.$ As
mentioned above, $W_{j}$ can be viewed as obtained from $V_{j}^{\text{ }}$by
multiplying one non-zero coordinate of $V_{j}$ by $2$. It is easy to see that,
for $1\leq j<k\leq4,$ if $W_{j}$ and $W_{k}$ were obtained by multiplying the
same coordinate of $V_{j}$ and $V_{k}$ by $2$ than their distance would be
less than $3$. On the other hand, if, for some $j,1\leq j\leq4,$ the fourth
coordinate of $W_{j}$ equaled to $-2$, then the distance of $W_{j}$ to $B$
would be less than three. Thus in this case, codewords covering $V_{j}%
,j=1,...,4,$ do not exist.\bigskip

(ib) In this case $\mathcal{H}=\pm<(-1,1,1,1,0)>,$ which in turn
implies\newline$\mathcal{H}^{\ast}-\mathcal{H}=\{\pm<(1,1,-1,1,0)>,\pm
<(1,1,-1,-1,0)>\},$ and\newline$\mathcal{G}^{\ast}\mathcal{=\{\pm
<}(1,1,0,1,0)>,\mathcal{\pm<}(1,1,-1,0,0)>,\pm$ $\mathcal{<}(1,-1,1,0,0)>\}.$
Let\medskip%

\begin{tabular}
[c]{rrrrrr}%
$B=($ & $2,$ & $1,$ & $0,$ & $0,$ & $0)$\\
$V_{1}=($ & $1,$ & $1,$ & $-1,$ & $1,$ & $0)$\\
$V_{2}=($ & $1,$ & $1,$ & $0,$ & $1,$ & $0)$\\
$V_{3}=($ & $1,$ & $1,$ & $-1,$ & $0,$ & $0)$\\
$V_{4}=($ & $0,$ & $1,$ & $-1,$ & $1,$ & $0)$%
\end{tabular}
\medskip

\noindent Then $V_{1}\in\mathcal{H}^{\ast}-\mathcal{H},$ $V_{2},V_{3},V_{4}%
\in\mathcal{G}^{\ast}.$ By the same type of an argument as in (ia) it is easy
to see that the required codewords $W_{j},j=1,...,4,$ do not exist, as
multiplying the first coordinate of $V_{j}$ by $2$ leads to a codeword at
distance less than $3$ from $B$. Moreover, our example shows that in this case
(ib) it is impossible to choose the sets of codewords of type $[2^{1},
1^{2}],$ and $[2^{1},1^{3}],$ regardless whether $<(-2,-1,0,0,0)>,$ or
$<(-1,-2,0,0,0)>$ is in $\mathcal{B}$.

\bigskip

\noindent(iia) We show that also in this case it is impossible to choose the
sets of codewords of type $[2^{1}, 1^{2}],$ and $[2^{1},1^{3}].$ There are
two ways how to choose codewords in $\mathcal{B}$ of type $[2^{1}, 1^{1}]$
covering words $<(-1,0,-1,0,0)>;$ either $<(-2,0,-1,0,0,)>,$ or
$<(-1,0,-2,0,0,)>.$ There are four ways how to choose the set $\mathcal{H}$ of
codewords of type $[1^{4}].$ We treat here only one of them as the other
three are nearly identical to this one. We deal with the option $\mathcal{H}%
=H_{2}\cup H_{3}.$ Then $\mathcal{H}^{\ast}-\mathcal{H=\{}%
<(1,1,1,-1,0)>,\newline<-1,-1,1,-1,0)>$ $,<(1,1,-1,0,-1)>,<(1,1,0,-1,-1)>\},$
while\newline$\mathcal{G}^{\ast}%
=\{<(1,0,1,0,1)>,<(1,0,-1,1,0)>,<(1,0,0,-1,-1)>,\newline%
<(1,1,0,0,-1)>,<(-1,-1,0,0,-1)>,<(-1,0,-1,0,1)>\}.\ $\ Of the two cases when
$<(-2,0,-1,0,0,)>\in\mathcal{T}_\mathcal{L}$ or $<(-1,0,-2,0,0,)>\in\mathcal{T}_\mathcal{L}$ \ we
treat here the first one. Let\medskip

$%
\begin{tabular}
[c]{rrrrrr}%
$B_{1}=($ & $2,$ & $1,$ & $0,$ & $0,$ & $0)$\\
$B_{2}=($ & $0,$ & $0,$ & $-2,$ & $0,$ & $-1)$\\
$V_{1}=($ & $1,$ & $1,$ & $-1,$ & $0,$ & $-1)$\\
$V_{2}=($ & $0,$ & $1,$ & $-1,$ & $0,$ & $-1)$\\
$V_{3}=($ & $1,$ & $1,$ & $0,$ & $0,$ & $-1)$%
\end{tabular}
$,\medskip

\noindent where $B_{1}$ and $B_{2}$ are codewords and $V_{1}\in\mathcal{H}%
^{\ast}-\mathcal{H}$, $V_{2}\in<(-1,0,-1,0,1)>$\linebreak$\in\mathcal{G}%
^{\ast}$, $V_{3}\in<(1,1,0,0,-1)>\in\mathcal{G}^{\ast}$. Because of $B_{1}$
none of $V_{j}^{\prime}$s can have the first coordinate multiplied by $2,$
while because of $B_{2}$ none of $V_{j}^{\prime}$s can have its third
coordinate multiplied by $2.$ Thus, codewords covering $V_{1},V_{2},V_{3}$ do
not exist.\bigskip

\noindent Thus, in what follows, it suffices to consider the case (ia) with
$\mathcal{B}=\pm<(2,1,0,0,0)>.$\bigskip

\noindent Now it is relatively simple to show the uniqueness of the codewords
of the remaining types.

\underline{Codewords of type $[3^{1},1^{1}].$} In $Z^{5}$ there are eighty words
of type $[2^{1}, 1^{1}].$ By (\ref{b}), ten of them are covered by the codewords
of type $[2^{1}, 1^{1}]$, and sixty by the codewords of type $[2^{1}, 1^{2}].$
The remaining ten words $\pm<(2,0,0,-1,0)>$ are to be covered by codewords of
type $[3^{1},1^{1}].$ Thus, $\pm<(3,0,0,-1,0)>$ are in $\mathcal{T}_\mathcal{L}$.

\underline{Codewords of type $[3^{1},1^{2}]$ and $[2^{2}, 1^{1}].$} In
$Z^{5\text{ }}$there are $240$ words of type $[2^{1}, 1^{2}].$ By (\ref{g})
sixty of them are covered by codewords of type $[2^{1}, 1^{1}],$ thirty by
codewords of type $[1^{3}],$ another thirty by codewords of type $[
2^{1},1^{2}]$, and sixty by codewords of type $[2^{1},1^{3}].$ The remaining
sixty words of type $[2^{1}, 1^{2}]$ are $\pm<(2,0,0,1,1)>,$\linebreak%
$\pm<(2,-1,0,0,-1)>,$\ then $\pm<(2,-1,0,0,1)>,$ $\pm<(1,-2,0,0,1)>,$ and
$\pm<(2,0,-1,0,-1)>,$ $\pm<(1,0,-2,0,-1)\dot{>}.$ As each codeword of type
$[2^{2}, 1^{1}]$ covers two words of type $[2^{1}, 1^{2}]\,$\ there have to
be in $\mathcal{T}_\mathcal{L}$ codewords $\pm<(2,-2,0,0,1)>$ and $\pm<(2,0,-2,0,-1)>$ of
type $[2^{2}, 1^{1}],$ and the remaining twenty words of type $[2^{1},
1^{2}]$ are covered by codewords $\pm<(3,0,0,1,1)>,\pm<(3,-1,0,0,-1)>$ of type
$[3^{1},1^{2}]. $

\underline{Codewords of type $[3^{1},2^{1}]$ and $[4^{1},1^{1}].$} Finally, the
ten remaining words $\pm<(3,0,-1,0,0)>$ of type $[3^{1},1^{1}]$ have to be
covered by codewords \newline$\pm<(4,0,-1,0,0)>$ of type $[4^{1},1^{1}]\,,$ and
the ten remaining words\newline$\pm<(2,0,2,0,0)>$ of type \ $[2^{2}]$ by
codewords\newline$\pm<(3,0,2,0,0)>$ of type $[3^{1},2^{1}].$\bigskip

\noindent So we have proved that the $5$-neighborhood of each point is
congruent to the canonical neighborhood, and that the neighborhood is uniquely
determined by the codewords of type $[2^{1}, 1^{1}].$\bigskip
\end{proof}

\ \noindent At the end of this subsection we describe an important attribute
of the canonical $5$-neighborhood.

\begin{theorem}
\label{13} Let $U,Z\in\mathcal{T}_\mathcal{L}$ be two words from the $5$-neighborhood of a
codeword $W$. Then $2W-U$ is in this neighborhood, that is, the $5$%
-neighborhood is symmetric, and if $\left\vert U+Z-2W\right\vert \leq5,$ then
$U+Z-W$ belongs to this $5$-neighborhood as well. In particular, if $U,Z $ are
from the $5$-neighborhood of the origin then $-U$ and $U+Z,$ if $\left\vert
U+Z\right\vert _{M}\leq5,$ belong to this neighborhood as well.
\end{theorem}

\begin{proof}
Again it suffices to prove the statement for $W=O.$ From the previous theorem
we know, that the $5$-neighborhood of each codeword in congruent to the
canonical one. Clearly, a congruence mapping retains the properties described
in this theorem. Therefore, it suffices to prove the statement for the
canonical neighborhood. To show that the canonical neighborhood satisfies
these properties we prove that this neighborhood is a part of (the unique)
lattice tiling of $R^{5}$ by crosses. Then the proof will follow from the fact
that if words $U,V$ belong to a lattice $\mathcal{L}$ then also $-U$ and $U+V$
$\ $are in $\mathcal{L}$.\bigskip

\noindent Consider a homomorphism $\phi:Z^{5}\rightarrow Z_{11},$ the cyclic
group of order $11,$ given by $\phi(e_{1})=1,\phi(e_{2})=9,\phi(e_{3}%
)=4,\phi(e_{4})=3,$ and $\phi(e_{5})=5.$ Then $\phi$ satisfies the assumptions
of Corollary \ref{20}. Thus, $\phi$ induces a lattice tiling $\mathcal{L}$ of
$R^{5}$ by crosses, where the set $\ker\phi$ is the set $\mathcal{T}%
_{\mathcal{L}}$ of the centers of crosses in this tilling. It is easy, although time consuming, to check, that
all codewords from the canonical neighborhood belong to $\ker\phi,$ that is,
if $U=(u_{1},...,u_{5})$ belongs to the $5$-neighborhood then $u_{1}%
+9u_{2}+4u_{3}+3u_{4}+5u_{5}=0(\operatorname{mod}11).$ The proof is complete.
\end{proof}

\subsection{Phase D}

\noindent As the closing part of the proof of Theorem \ref{4} we show that for
any two codewords in $\mathcal{T}_\mathcal{L}$ their $5$-neighborhoods are not only
congruent but that they are identical.

\begin{theorem}
The $5$-neighborhood of each codeword in the tiling $\mathcal{L}$ is equal to
the $5$-neighborhood of the origin.
\end{theorem}

\begin{proof}
Let $\mathcal{L}$ be a tiling of $R^{5}$ by crosses. We proved that the
$5$-neighbor\-hood of any codeword $W$ in $\mathcal{T}_\mathcal{L}$ is congruent to the
$5$-neighborhood of the origin. By Claim \ref{0}, we may assume that the
$5$-neighborhood of the origin is the canonical one.\bigskip

\noindent We have proved, see Theorem \ref{6}, that for each codeword $W$ the
$3$-neighbor\-hood of $W$ comprises twenty codewords; i.e., there are in
$\mathcal{T}_\mathcal{L}$ twenty codewords at distance $3$ from $W$. We will call these
codewords at distance $3$ from $W$ the codewords adjacent to $W.$ Ten of the
adjacent codewords are of type $[2^{1}, 1^{1}],$ and ten of them are of type
$[1^{3}].$ The proof of the theorem is based on the following claim, which
states that all codewords adjacent to $W$ have "the same" set of codewords of
type $[2^{1}, 1^{1}]$ as $W$ has.\bigskip

\noindent\textbf{Claim D. }Let $W$ be a codeword in $\mathcal{T}_\mathcal{L}$, and let $U$
be a codeword adjacent to $W.$ Further, let $S_{W},$ and $S_{U}$ be the set of
codewords of type $[2^{1}, 1^{1}]$ with respect to $W$ and $U,$ respectively.
Then $\{Z-W;$ where $Z\in S_{W}\}=\{Z-U,$ where $Z\in S_{U}\}.$\bigskip

\noindent Theorem \ref{12} states that the $5$-neighbourhood of each codeword
$W$ is uniquely determined by the codewords of type $[2^{1}, 1^{1}]$. Thus, with
Claim D in hands, we know that any two adjacent codewords have the same
$5$-neighborhood. The rest of the proof of the theorem follows easily by
induction because to each codeword $W$ there is a sequence of codewords
$O=Z_{0},Z_{1},...,Z_{m-1},Z_{m}\\=W$ such that the codeword $Z_{j}$ is adjacent
to the codeword $Z_{j-1}$ for all $j=1,...,m.$\bigskip

\noindent W.l.o.g. we prove Claim D only for $W=O,$ the origin. Consider a
codeword $U$ that is adjacent to $O.$ To prove Claim D for $U$ we need to show
that $U+V,$ where $V\in\pm<(2,1,0,0,0)>$ is a codeword in $\mathcal{T}_\mathcal{L} $. To
do so, it suffices either

(a) to show that $U-V$ is a codeword, or

(b) to choose $X,Y,Z$ so that

(i) $X,Y,Z,Y-X,Z-X,$ and $Y+Z-2X$ are in the canonical neighborhood; and

(ii) $Y+Z-X=U+V.$

\noindent Indeed, in the case (a), we know that for each codeword $W$ its
$5$-neighborhood is symmetric with respect $W,$ hence if $U-Z$ is a codeword
then also $U+Z$ is a codeword because $\left\vert Z\right\vert _{M}\leq5.$ In
the case (b) consider a codeword $X$. By Theorem \ref{13}, if there are
codewords $Y,Z,$ so that $\rho_{M}(Y,X)\leq5,\rho_{M}(Z,X)\leq5,$ and
$\rho_{M}(Y+Z-X,X)\leq5,$ then $Y+Z-X$ is a codeword as well. However, (i)
guarantees that all assumptions of Theorem \ref{13} are satisfied, therefore
(ii) guarantees that $U+V$ is in the $5$-neighborhood of the codeword
$U$.\bigskip

\noindent First we choose the codeword adjacent to the origin to be of type
$[2^{1}, 1^{1}].$ Let $U=(2,1,0,0,0).$ If $V=(2,1,0,0,0)$ then $U-V=(0,...,0)$ is
a codeword, and hence by (a) $U+V=(4,2,0,0,0)$ is a codeword as well. The
following table provides a suitable choice for the other four codewords in
$<(2,1,0,0,0)>.$

\noindent\vspace{-0.7cm}
\begin{center}
\scalebox{0.87}{
\begin{tabular}
[c]{|l|l|l|l|l|}\hline
$\ \ \ \ \ \ \ V$ & $\ \ \ \ \ \ \ X$ & $\ \ \ \ \ \ \ Y$ & $\ \ \ \ \ \ \ Z$
& $Y+Z-X=U+V$\\\hline
$(0,2,1,0,0)$ & $(1,0,1,0,-1)$ & $(0,3,0,0,-1)$ & $(3,0,2,0,0)$ &
$(2,3,1,0,0)$\\\hline
$(0,0,2,1,0)$ & $(2,1,0,0,0)$ & $(1,2,0,1,0)$ & $(3,0,2,0,0)$ & $(2,1,2,1,0)$%
\\\hline
$(0,0,0,2,1)$ & $(2,1,0,0,0)$ & $(1,2,0,1,0)$ & $(3,0,0,1,1)$ & $(2,1,0,2,1)$%
\\\hline
$(1,0,0,0,2)$ & $(0,1,0,-1,1)$ & $(3,0,0,-1,0)$ & $(0,2,0,0,3)$ &
$(3,1,0,0,2)$\\\hline
\end{tabular}
}
\end{center}

\noindent Theorem \ref{13} guarantees that the $5$-th neighborhood of each
codeword is symmetric. Therefore $U+V$ is a codeword also for all \newline%
$V\in <(-2,-1,0,0,0)>.$ Let $U^{\prime}$ $\in<(2,1,0,0,0)>.$ Then we apply the
same cyclic shift to $X$,$Y,Z$ given in the table to obtain the required
codewords. The same applies to $-U$ and its cyclic shifts. Finally, let $U$ be
a codeword of type $[1^{3}]$ adjacent to the origin. say $U=(1,0,1,0,-1).$
The proper choice of $X,Y,Z$ for each $V\in<(2,1,0,0,0)>$ is given in the
table below:\medskip

\noindent\vspace{-1cm}
\begin{center}
\scalebox{0.86}{
\begin{tabular}
[c]{|l|l|l|l|l|}\hline
$\ \ \ \ \ \ \ \ V$ & $\ \ \ \ \ \ \ X$ & $\ \ \ \ \ \ \ \ Y$ &
$\ \ \ \ \ \ \ \ Z$ & $Y+Z-X=U+V$\\\hline
$(2,1,0,0,0)$ & $(1,0,1,0,-1)$ & $(1,1,0,0,-2)$ & $(3,0,2,0,0)$ &
$(3,1,1,0,-1)$\\\hline
$(0,2,1,0,0)$ & $(0,0,2,1,0)$ & $(0,1,1,1,-1)$ & $(1,1,3,0,0)$ &
$(1,2,2,0,-1)$\\\hline
$(0,0,2,1,0)$ & $(0,0,2,1,0)$ & $(1,0,1,2,0)$ & $(0,0,4,0,-1)$ &
$(1,0,3,1,-1)$\\\hline
$(0,0,0,2,1)$ & $(2,1,0,0,0)$ & $(2,0,1,0,1)$ & $(1,1,0,2,-1)$ & $(1,0,1,2,0)
$\\\hline
$(1,0,0,0,2)$ & $(2,1,0,0,0)$ & $(3,0,0,-1,0)$ & $(1,1,1,1,1)$ & $(2,0,1,0,1)
$\\\hline
\end{tabular}
}
\end{center}
\medskip

\noindent while for $V\in<(-2,-1,0,0,0)>$, $U+V$ is a codeword as the
$5$-neighbour\-hood is symmetric. As in the case of $U$ being an adjacent
codeword of type $[2^{1}, 1^{1}],$ a suitable choice of $X,Y,Z$ for other cases
can be obtained by a cyclic shift. The proof is complete.
\end{proof}

\begin{acknowledgement}
The authors thank Prof. Alex Rosa for checking correctness of the main result
of the paper.
\end{acknowledgement}

\end{document}